\documentclass[12pt]{article}

\usepackage{amsfonts}
\usepackage{amssymb}
\usepackage{amsmath}
\usepackage{amsthm,color}

\usepackage{verbatim}
\allowdisplaybreaks

\usepackage[top=1in, bottom=1in, left=1in, right=1in]{geometry}

\newtheorem{theorem}{Theorem}
\newtheorem{corollary}[theorem]{Corollary}
\newtheorem{rem1}[theorem]{Remark}
\newtheorem{lemma}[theorem]{Lemma}
\newtheorem{definition}[theorem]{Definition}
\newtheorem{proposition}[theorem]{Proposition}

\newtheorem{ex1}[theorem]{Example}
\newtheorem{ass1}[theorem]{Assumption}

\newenvironment{remark}{\begin{rem1}\rm}{\end{rem1}}

\newenvironment{assumption}{\begin{ass1}\rm}{\end{ass1}}

\numberwithin{equation}{section}
\numberwithin{theorem}{section}

\newcommand{\N}{\mathbb{N}}
\renewcommand{\P}{\mathbb{P}}
\newcommand{\Q}{\mathbb{Q}}
\newcommand{\R}{\mathbb{R}}

\newcommand{\W}{\mathcal{W}}

\newcommand{\hatW}{\W}

\newcommand{\QQ}{\mathcal{Q}}

\renewcommand{\SS}{\mathcal{S}}
\newcommand{\Y}{\mathcal{Y}}
\newcommand{\K}{\mathbb{K}}
\newcommand{\zero}{{\bf 0}}

\newcommand{\lrparen}[1]{\left(#1\right)}
\newcommand{\lparen}[1]{\left(#1\right.}
\newcommand{\rparen}[1]{\left.#1\right)}
\newcommand{\lrsquare}[1]{\left[#1\right]}

\newcommand{\lrcurly}[1]{\left\{#1\right\}}

\newcommand{\Ft}[1]{\mathcal{F}_{#1}}
\newcommand{\Lnp}[3]{L_{#1}^{#2}(#3)}

\newcommand{\LdpF}[1]{\Lnp{#1}{p}{\R^d}}
\newcommand{\LdqF}[1]{\Lnp{#1}{q}{\R^d}}

\newcommand{\LdiF}[1]{\Lnp{#1}{\infty}{\R^d}}

\newcommand{\LdoF}[1]{\Lnp{#1}{1}{\R^d}}

\newcommand{\LdzF}[1]{\Lnp{#1}{0}{\R^d}}
\newcommand{\LdpK}[3]{\Lnp{#2}{#1}{#3}}
\newcommand{\LdiK}[2]{\Lnp{#1}{\infty}{#2}}
\newcommand{\LdoK}[2]{\Lnp{#1}{1}{#2}}

\newcommand{\LpF}[1]{\Lnp{#1}{p}{\R}}

\newcommand{\LiF}[1]{\Lnp{#1}{\infty}{\R}}
\newcommand{\LoF}[1]{\Lnp{#1}{1}{\R}}
\newcommand{\LzF}[1]{\Lnp{#1}{0}{\R}}
\newcommand{\LpK}[3]{\Lnp{#2}{#1}{#3}}

\newcommand{\E}[1]{\mathbb{E}\lrsquare{#1}}
\newcommand{\EP}[2]{\mathbb{E}^{#1}\lrsquare{#2}}
\newcommand{\EQ}[1]{\EP{\Q}{#1}}

\newcommand{\EtQ}[1]{\EP{\tilde\Q}{#1}}

\newcommand{\Et}[2]{\E{\left.#1 \right| \mathcal{F}_{#2}}}
\newcommand{\EPt}[3]{\EP{#1}{\left.#2 \right| \mathcal{F}_{#3}}}
\newcommand{\EQt}[2]{\EPt{\Q}{#1}{#2}}
\newcommand{\EbQt}[2]{\EPt{\bar{\Q}}{#1}{#2}}
\newcommand{\EtQt}[2]{\EPt{\tilde\Q}{#1}{#2}}
\newcommand{\EQnt}[3]{\EPt{\Q_{#1}}{#2}{#3}}

\newcommand{\dQdP}{\frac{d\Q}{d\P}}
\newcommand{\dQnmdP}[2]{\frac{d\Q_{#1}^{#2}}{d\P}}
\newcommand{\dQndP}[1]{\dQnmdP{#1}{}}

\newcommand{\dRndP}[1]{\frac{d\R_{#1}}{d\P}}
\newcommand{\dRidP}{\dRndP{i}}

\newcommand{\genseq}[4]{\lrparen{#1_#4}_{#4=#2}^{#3}}
\newcommand{\seq}[1]{\genseq{#1}{0}{T}{t}}

\newcommand{\trans}[1]{#1^{\mathsf{T}}}
\newcommand{\transp}[1]{\trans{\lrparen{#1}}}
\newcommand{\prp}[1]{#1^{\perp}}

\newcommand{\plus}[1]{#1^+}

\newcommand{\diag}[1]{\operatorname{diag}\lrparen{#1}}
\newcommand{\cl}{\operatorname{cl}}
\newcommand{\ascl}{\operatorname{cl}}
\newcommand{\co}{\operatorname{co}}
\newcommand{\recc}[1]{\operatorname{recc}\lrparen{#1}}

\newcommand{\interior}{\operatorname{int}}

\newcommand{\abs}[1]{\left|#1\right|}

\newcommand{\as}{\text{a.s.}}
\newcommand{\sas}{\text{ a.s.}}

\DeclareMathOperator*{\esssup}{ess\,sup}
\DeclareMathOperator*{\essinf}{ess\,inf}

\begin{document}
\title{Scalar multivariate risk measures with a single eligible asset}
\author{Zachary Feinstein \thanks{Stevens Institute of Technology, School of Business, Hoboken, NJ 07030, USA, zfeinste@stevens.edu.} \and Birgit Rudloff \thanks{Vienna University of Economics and Business, Institute for Statistics and Mathematics, Vienna A-1020, AUT, brudloff@wu.ac.at.}}
\date{\today~(Original:~July 27, 2018)}
\maketitle
\abstract{
In this paper we present results on scalar risk measures in markets with transaction costs.  Such risk measures are defined as the minimal capital requirements in the cash asset.  First, some results are provided on the dual representation of such risk measures, with particular emphasis given on the space of dual variables as (equivalent) martingale measures and prices consistent with the market model.
Then, these dual representations are used to obtain the main results of this paper on time consistency for scalar risk measures in markets with frictions. It is well known from the superhedging risk measure in markets with transaction costs, as in \cite{JK95,RZ11,LR11}, that the usual scalar concept of time consistency is too strong and not satisfied. 
We will show that a weaker notion of time consistency can be defined, which corresponds to the usual scalar time consistency but under any fixed consistent pricing process. We will prove the equivalence of this weaker notion of time consistency and a certain type of backward recursion with respect to the underlying risk measure with a fixed consistent pricing process. Several examples are given, with special emphasis on the superhedging risk measure.
}

\section{Introduction}
\label{sec_intro}

The dynamic programming principle and time consistency are vital concepts both for the computation of dynamic optimization problems and guaranteeing that decisions made will not be reversed by the same decision maker at the next point in time.  In the context of risk management, these principles are encoded by properties on dynamic risk measures.  Such functions provide the minimal capital requirements necessary to make an investment acceptable to a risk manager.  In this paper we investigate time consistency properties for multivariate risk measures to allow for the possibility of transaction and market impact costs. That is, the frictionless conversion between assets is not possible and therefore the reduction of a portfolio to a single num\'eraire is either non-unique or implicitly requires a systematic increase in riskiness.  As will be detailed later in this work, our approach reduces this problem to the time consistency of the dynamic risk measures under all possible no-arbitrage pricing processes.

The organization of this paper is as follows. In Section~\ref{sec_intro-lit}, we provide a literature review  and in Section~\ref{sec_intro-motiv} an overview of the main results along with our motivation for their study.  In Section~\ref{sec_prelim} we present background material on risk measures and notation that will be used throughout this paper.  Much of this notation is comparable to that utilized in the set-valued risk measure literature.  In Section~\ref{sec_scalar} we present the definition of the dynamic scalar risk measures we will consider in this paper and basic results on boundedness properties.  In Section~\ref{sec_e1} we present dual representations for the scalar risk measures with a single eligible asset.  Special emphasis is placed on a representation akin to \eqref{JK d assets} below considered in~\cite{JK95}, but for general convex and coherent dynamic risk measures in markets with transaction costs, which relates those risk measures to the set of no-arbitrage pricing processes.  In Section~\ref{sec_relevance} we present results on relevance or sensitivity of the scalar risk measures which provide a condition for the dual representation to be with respect to equivalent probability measures only.  In Section~\ref{sec_tc} we introduce a new notion of time consistency for scalar risk measures with a single eligible asset.  This notion coincides with time consistency of the corresponding univariate scalar risk measure under any market-consistent frictionless (and no-arbitrage) price process.  As such, unlike much of the prior literature, we can and do demonstrate that such a property is satisfied by the usual examples; in particular we give details on the superhedging risk measure and composed risk measures in Section~\ref{sec_ex}.

\subsection{Literature review}\label{sec_intro-lit}

In their seminal paper, Artzner, Delbaen, Eber, and Heath \cite{AD99} introduced the concept of coherent risk measures in a static univariate setting which provides the minimal capital necessary to compensate for the risk of a contingent claim.  Coherent risk measures were further studied in~\cite{D02}.  Such risk measures were generalized to the convex case in~\cite{FS02,FG02} while retaining the same financial interpretation and a notion of diversification.

When a filtration $\seq{\mathcal{F}}$ is introduced, it is natural to consider dynamic risk measures, i.e., the minimal capital necessary to compensate for the risk of a contingent claim conditionally on the information at time $t$.  In this time-dynamic setting, the manner in which the risk of a contingent claim propagates through time is of great importance, as it has significant implications on risk management.  One such condition, called (strong) time consistency, consists of the condition that if one portfolio is riskier than another in the future, that same ordering must hold at all prior times as well.  This property is studied in the univariate setting in~\cite{AD07,R04,DS05,CDK06,RS05,BN04,BN08,BN09,FP06,CS09,CK10,AP10,FS11} in discrete time and~\cite{FG04,D06,DPRG10} in continuous time.  Other, weaker, forms of time consistency have been presented for dynamic risk measure in, e.g.,~\cite{lacker2018liquidity,lacker2018law,roorda2007time,tutsch2008update,weber2006distribution}.

In this work, we will consider multivariate risk measures.  The multivariate setting arises from markets with frictions, as the liquidation of a portfolio into some num\'eraire does not allow for repurchasing the same asset.  This setting was studied in a set-valued static (one period) framework in~\cite{JMT04,HHH07,HR08,HH10,HHR10}.  Set-valued risk measures, in the one period framework, have recently been applied to studying systemic risk in~\cite{FRW15,BFFM15}.  In a set-valued time-dynamic framework a new notion of time consistency, called multiportfolio time consistency, was introduced in \cite{FR12,TL12}.  This property was further studied in \cite{FR12b,FR13-survey,FR15-supermtg,FR18-scalar,CH17}.  Computation of such set-valued risk measures was studied in~\cite{FR14-alg,LR11}.

The focus of this work is on dynamic scalar multivariate risk measures with a single eligible asset.  Some results on dynamic scalar multivariate risk measures with multiple eligible assets have been discussed previously in~\cite{FR13-survey,FR15-supermtg,FR18-scalar}.  
Scalar risk measures in frictionless markets with either a single eligible asset \cite{FS02,AD99} or multiple eligible assets \cite{FMM13,BMM18,WH14,ADM09,FS06,K09,Sc04} can be considered as scalar multivariate risk measures (see \cite[Example~2.26]{FR13-survey} for such a comparison).
Additionally, such functions have been considered in the context of systemic risk in, e.g., \cite{KOZ16,BC13}.
Our focus in this paper is of the dynamic version of the multivariate risk measures with a single eligible asset as presented in \cite{BR06,WA13} (see \cite[Example~2.27 and~2.28]{FR13-survey} for a brief discussion).

Of course time (in-)consistency has not just been studied in the context of risk measurement as considered here. There has been substantial development and investigation of time (in-)consistency for stochastic control and stopping problems coming, e.g., from behavioral economics where one considers the (in-)consistency of an agent's control (i.e., the choice of strategies or stopping times) over time.  Time (in-)consistency was surveyed for utility maximization problems by~\cite{strotz1955myopia}. We wish to highlight the recent works of~\cite{ekeland2008investment,bjork2014mean,bjork2017time} on time consistency for stochastic control and~\cite{huang2018time,christensen2018finding} on time consistency for optimal stopping problems; these works utilize a standard equilibrium or consistent planning approach proposed in~\cite{strotz1955myopia}.  We also wish to highlight~\cite{KMZ17} which very succinctly defines such problems of time (in-)consistency for control problems; that work proposes and uses several different approaches distinct from the equilibrium approach of the prior references.  Such problems in time (in-)consistency have drawn increased attention since the seminal work of Ekeland and Lazrak~\cite{ekeland2006being}.
In fact, such time (in-)consistency questions are intimately related with the same notions for risk measures.  This control viewpoint can be seen through the dual representation, which is a maximization problem over probability measures.  This is made explicit in the coherent case via stability~\cite{AD07}, i.e., the pasting of these probability measures through time.  Indeed, we can view time consistency of risk measures as the recursion of the value function from the control literature.

The situation is similar in the set-valued framework. Stochastic control problems in a multivariate setting, where the value function is now a set-valued function, have been studied in \cite{KR18,FRZ20}; the corresponding backward recursion of the value function in the time consistent case is in analogy to the backward recursion of set-valued risk measures in the multi-portfolio time consistent case, as the value function (i.e., the risk measure itself) recurses backwards in time. Of course stochastic control problems can have a much more complicated structure as, e.g., dependencies on initial capital levels can complicate the numerical computations in the backward recursion, see \cite{KR18}, compared to risk measures, where the input is the random vector at terminal time.

\subsection{Motivation}\label{sec_intro-motiv}

In this work we consider risk measurement in markets with transaction costs.  These market frictions result in the need to consider multivariate risk measures (i.e., risk measures for a vector-valued portfolio of ``physical units'' -- the number of units of the $d$ available assets or securities held -- rather than the value of that portfolio in a single num\'eraire) because no single value of a portfolio or investment can be defined.  For instance, a mark-to-market value of the portfolio using the current mid-price between the bid and ask prices is strictly larger than the liquidation value of that same portfolio; both such values would result in different capital requirements and neither, necessarily, encodes the true risk of the portfolio.  Set-valued risk measures (see, e.g.,~\cite{JMT04,HHR10}) were introduced in order to handle this problem by considering the set of all capital requirements, possibly held in multiple different assets, that make the portfolio acceptable. This eliminates the need to first convert the $d$-dimensional portfolio vector into a one-dimensional quantity and thus avoids the problems related to that.

We are focused on dynamic multivariate risk measures and questions surrounding time consistency.  For set-valued risk measures, the typical dynamic programming principle is called multiportfolio time consistency (see, e.g.,~\cite{FR12,FR12b}).  Such a property relates the risk between vector-valued portfolios over time by comparing the full set of acceptable capital allocations.  The usual risk measures, e.g., the superhedging risk measure, satisfy this time consistency property when the full space of assets are eligible to be held as capital.  However, typically, capital is required to be held in a num\'eraire, or cash, asset only.  
In such a setting, multiportfolio time consistency may no longer hold for the corresponding set-valued risk measure (this is the case, e.g., for superhedging). Furthermore, for a single eligible asset, the minimal capital requirement (often called the scalarization of the set-valued risk measure) becomes the logical object to consider, but it typically does not satisfy the scalar concept of time consistency either.

So the motivating questions of this paper are the following. If the eligible asset is just the cash asset, the set-valued risk measure does not satisfy multiportfolio time consistency, and the scalarization does not satisfy the usual time consistency, is there a more appropriate, weaker form of time consistency for this scalar risk measure? Can one use the fact that the set-valued risk measure, for the full set of eligible assets, is multiportfolio time consistent? The answer to both questions is yes and will be answered in this work, see in particular Lemma~\ref{lemma:mptc-pi_tc} below.

We are specifically motivated to study the problem of scalar multivariate risk measures with a single eligible asset from the many works that consider the scalar superhedging price in a market with frictions.
We refer to, e.g., \cite{BLPS92,BV92,PL97,JK95,R08} for studies of the scalar superhedging price in a market with transaction costs and two assets. This has been extended to a general number of assets in \cite{RZ11,LR11}. The $d$ dimensional version of the scalar superhedging price given by \cite{JK95} under a sequence $\seq{K}$ of solvency cones (for a market with proportional transaction costs, see Section~\ref{sec_scalar} for further discussion) is as follows.
Under the appropriate no-arbitrage argument, the scalar superhedging price $\rho^{SHP}(-X)$ in units of the first (cash) asset at time $t=0$ is given by
\begin{align}
    \label{JK d assets}
   \rho^{SHP}(-X)= \sup_{(\Q,S)\in\mathcal Q} \EQ{\trans{S_T}X}
\end{align}
for $X \in L^\infty(\Omega,\mathcal F_T,\P;\R^d)$ which is the payoff in physical units. The set of dual variables is given by $\mathcal Q$ which denotes the set of all processes $S = \seq{S}$ along with their equivalent martingale measures $\Q$ such that $S_t^1\equiv 1$ and $\Et{\frac{d\Q}{d\P}}{t} S_t\in L^\infty(\Omega,\Ft{t},\P;\plus{K_t})$  for all $t$.
For the context of this paper, we note that Theorem~6.1 of \cite{LR11} relates~\eqref{JK d assets} to the scalarization of the set-valued superhedging risk measure under a single eligible asset (the num\'{e}raire asset).
We will present a similar dual representation for general convex risk measures in a dynamic framework.  We then use this representation to define a new time consistency property, as detailed below, which is satisfied by, e.g., the superhedging risk measure.  This is in contrast to prior works on time consistency for scalar multivariate measures \cite{HMBS16,KOZ15} whose time consistency property does not include the superhedging risk measure as an example. This agrees also with recent research on scalarizations of multivariate problems, as in \cite{KMZ17,KR18}, indicating that the usual scalar concept of time consistency is way too strong to be satisfied for scalarizations of multivariate risk measures.
Though weaker forms of time consistency (as in, e.g.,~\cite{weber2006distribution}) may be satisfied for the scalarization, we are interested in a weaker version that is strong enough to still allow a relation to the dynamic programming principle.

We call the weaker time consistency property defined in this work $\pi$-time consistency as it corresponds to the usual scalar time consistency for all the risk measure $\pi^S_t$ appearing when fixing a pricing process $S$.  That is, the risk measure is time consistent in the usual way when the market model follows any (frictionless) no-arbitrage pricing process $S$. We will prove the equivalence of this weaker notion of time consistency and a certain type of backward recursion with respect to $\pi^S_t$. We will show that the superhedging risk measure in markets with transaction costs (expressed in the cash asset) as in \cite{JK95,RZ11,LR11} while not satisfying the usual scalar time consistency property, which also means it is not recursive with itself, will satisfy  $\pi$-time consistency and thus a backward recursion with respect to $\pi^S_t$.

\section{Setup}
\label{sec_prelim}
Consider a filtered probability space $\lrparen{\Omega,\Ft{},\seq{\mathcal{F}},\P}$ satisfying the usual conditions with $\Ft{0}$ the trivial sigma algebra and $\Ft{} = \Ft{T}$. Let $|\cdot|_n$ be an arbitrary norm in $\R^n$ for $n \in \N$; note that $|\cdot|_1$ is equivalent to the absolute value operator.
Denote by $\LdpK{0}{t}{D} = \LdpK{0}{}{\Omega,\Ft{t},\P;D}$ the set of equivalence classes of $\Ft{t}$-measurable functions $X: \Omega \to D \subseteq \R^n$ for some $n \in \N$.  Additionally, denote by $\LdpK{p}{t}{D} \subseteq \LdpK{0}{t}{D}$ those random variables $X \in \LdpK{0}{t}{D}$ such that $\|X\|_p = \lrparen{\int_\Omega |X(\omega)|_n^p d\P}^{\frac{1}{p}} < +\infty$ for $p \in (0,+\infty)$ and $\|X\|_\infty = \esssup_{\omega \in \Omega} |X(\omega)|_n < + \infty$ for $p = +\infty$.  We further note that $\LdpK{p}{}{D} := \LdpK{p}{T}{D}$ by definition.  
In particular, if $p = 0$, $\LdzF{t}$ is the linear space of the
equivalence classes of $\Ft{t}$-measurable functions $X: \Omega \to \R^d$.  For $p > 0$, $\LdpF{t}$ denotes the linear space of
$\Ft{t}$-measurable functions $X: \Omega \to \R^d$ such that $\|X\|_p < +\infty$ for $p \in
(0,+\infty]$. 
Note that an element $X \in \LdpF{t}$ has components $X_1,...,X_d$ in $\LpF{t}$ for any choice of $p \in [0,+\infty]$.
For $p \in \lrsquare{1,+\infty}$ we will consider the dual pair $\lrparen{\LdpF{t}, \LdqF{t}}$, where $\frac{1}{p}+\frac{1}{q} = 1$ (with $q = +\infty$ when $p = 1$ and $q = 1$ when $p = +\infty$),  and endow it with the norm topology, respectively the $\sigma \lrparen{\LdiF{t},\LdoF{t}}$-topology on $\LdiF{t}$ in the case $p = +\infty$.
Throughout most of this paper we will focus on the case that $p = +\infty$.

Throughout this work we will make use of the indicator functions which we will denote by $1_D: \Omega \to \{0,1\}$ for some $D \in \Ft{}$.  These are defined so that $1_D(\omega) = 1$ if $\omega \in D$ and $0$ otherwise.
Additionally, throughout we will consider the summation of sets by Minkowski addition.
A set $B \subseteq \LdpF{}$ is called $\Ft{t}$-decomposable if $1_D B + 1_{D^c} B \subseteq B$ for any measurable $D \in \Ft{t}$.

As in \cite{K99} and discussed in \cite{S04,KS09}, the portfolios in this paper are in ``physical units'' of an asset (in a market with $d$ assets) rather than the value in a fixed num\'{e}raire, except where otherwise mentioned.  That is, for a portfolio $X \in \LdiF{t}$, the values of $X_i$ (for $1 \leq i \leq d$) are the number of units of asset $i$ in the portfolio at time $t$.

Let $M = \R \times \{0\}^{d-1}$ denote the set of eligible portfolios, i.e., those portfolios which can be used to compensate for the risk of a portfolio. That is, for the purpose of this paper, we will assume that the first asset is the only asset available to compensate for risk (in typical examples this will be the cash asset).  For notational simplicity we define $M_t = \LdiK{t}{M} = \LiF{t} \times \{0\}^{d-1}$, a weak* closed linear subspace of $\LdiF{t}$ (see Section~5.4 and Proposition~5.5.1 in \cite{KS09}).
We will denote $M_+ := M \cap \R^d_+ = \R_+ \times \{0\}^{d-1}$ to be the nonnegative elements of $M$.  We will additionally denote $M_{t,+} := M_t \cap \LdiK{t}{\R^d_+} = \LdiK{t}{M_+} = \LdiK{t}{\R_+} \times \{0\}^{d-1}$ to be the nonnegative elements of $M_t$.

Denote the upper sets by $\mathcal{P}\lrparen{M_t;M_{t,+}}$ where $\mathcal{P}\lrparen{\mathcal{Z};C} := \lrcurly{D \subseteq \mathcal{Z} \; | \; D = D + C}$ for some vector space $\mathcal{Z}$ and an ordering cone $C \subset \mathcal{Z}$.  Additionally, let $\mathcal{G}(\mathcal{Z};C) := \lrcurly{D \subseteq \mathcal{Z} \; | \; D = \cl\co\lrparen{D + C}} \subseteq \mathcal{P}(\mathcal{Z};C)$ be the upper closed convex subsets.

As noted previously, for the duality results below we will consider the weak* topology for $p = +\infty$.  We will briefly describe the set of dual variables from the set-valued biconjugation theory as utilized in, e.g., \cite{FR12,FR12b}, and first defined in \cite{H09}.
First, define the space of probability measures absolutely continuous with respect to $\P$ as $\mathcal{M}$ and those that are equivalent probability measures as $\mathcal{M}^e$.  The space of $d$-dimensional probability measures absolutely continuous with respect to $\P$ is thus denoted by $\mathcal{M}^d$.
Throughout we will use a $\P$-almost sure version of the $\Q$-conditional expectation where $\Q \in \mathcal{M}$.  This is defined in, e.g., \cite{CK10,FR12}.  Briefly, let $X \in \LiF{}$, then define the conditional expectation
\[\EQt{X}{t} := \Et{\bar{\xi}_{t,T}(\Q)X}{t},\]
where 
\[\bar{\xi}_{s,\sigma}(\Q)[\omega] := \begin{cases}\frac{\Et{\dQdP}{\sigma}(\omega)}{\Et{\dQdP}{s}(\omega)} & \text{if } \Et{\dQdP}{s}(\omega) > 0 \\ 1 &\text{else} \end{cases}\]
for every $\omega \in \Omega$.  For a vector-valued probability measure $\Q \in \mathcal{M}^d$ the result is defined component-wise, i.e., $\EQt{X}{t} = \transp{\EQnt{1}{X_1}{t},\dots,\EQnt{d}{X_d}{t}}$ for any $X \in \LdiF{}$, and define $\xi_{s,\sigma}(\Q) := \transp{\bar{\xi}_{s,\sigma}(\Q_1),\dots,\bar{\xi}_{s,\sigma}(\Q_d)}$.
Note that, for any $\Q \in \mathcal{M}$ and any times $0 \leq t \leq s \leq \sigma \leq T$, $\dQdP = \bar{\xi}_{0,T}(\Q)$ and $\bar{\xi}_{t,\sigma}(\Q) = \bar{\xi}_{t,s}(\Q) \bar{\xi}_{s,\sigma}(\Q)$ almost surely.

From this representation we can define the set of dual variables from the set-valued biconjugation theory \cite{H09} to be
\begin{equation}\label{eq_dualvar}
\W_t := \lrcurly{(\Q,w) \in \mathcal{M}^d \times \lrparen{\plus{M_{t,+}} \backslash \prp{M_t}} \; | \; w_t^T(\Q,w) \in \LdpK{1}{}{\R^d_+}}.
\end{equation}
In this representation we define \[w_t^s(\Q,w) = \diag{w}\xi_{t,s}(\Q)\] for any times $0 \leq t \leq s \leq T$ and where $\diag{x}$ denotes the diagonal matrix with the components of $x$ on the main diagonal.
Additionally, we employ the notation $\prp{M_t} = \lrcurly{v \in \LdoF{t} \; | \; \E{\trans{v}u} = 0 \; \forall u \in M_t}$ to denote the orthogonal space of $M_t$ and $C^+=\lrcurly{v \in \LdoF{t} \; | \; \E{\trans{v}u} \geq 0 \; \forall u \in C}$ to denote the positive dual cone of a cone $C\subseteq \LdiF{t}$.

\section{Scalarizations}

In this section we introduce the scalar multivariate risk measures of interest in this paper.  We provide, in Section~\ref{sec_scalar}, an axiomatic definition of these risk measures akin to those in~\cite{AD99,FS02,FG02} and, further, relate these scalar risk measures to the set-valued risk measures of~\cite{FR12,FR12b,FR13-survey}.  In Section~\ref{sec_e1}, we consider the dual representation of these scalar risk measures producing two, equivalent, representations. The first is with respect to the dual variables considered in the set-valued literature, the second providing a representation generalizing the dual representation \eqref{JK d assets} of the scalar multivariate superhedging risk measure from~\cite{JK95}.  We conclude with Section~\ref{sec_relevance} to extend the dual representation results in order to determine a sufficient condition for such a representation to be over the space of equivalent probability measures only.

\subsection{Definition}
\label{sec_scalar}

First, we wish to consider an axiomatic definition for the scalar conditional convex risk measures which will be the subject of this paper. These properties provide the classical interpretations that the risk measure is a capital requirement, that having more in each asset corresponds with lower risk, that diversification reduces risk, and that not being invested in the market carries with it a finite risk (possibly 0). Throughout this paper we will consider the unit vector $e_1 := \transp{1,0,...,0} \in \R^d$ and the zero vector $\zero := \transp{0,0,...,0} \in \R^d$.

\begin{definition}
\label{defn_scalar}
A mapping $\rho_t: \LdiF{} \to \LiF{t}$ is called a \textbf{\emph{scalarized conditional convex risk measure}} if it satisfies the following properties for all $X,Y \in \LdiF{}$ and $m \in \LiF{t}$:
\begin{itemize}
\item \emph{Conditional cash invariance}: $\rho_t(X + m e_1) = \rho_t(X) - m$;
\item \emph{Monotonicity}: $Y - X \in \LdiK{}{\R^d_+} \Rightarrow \rho_t(X) \geq \rho_t(Y)$;
\item \emph{Conditional convexity}: $\rho_t(\lambda X + (1-\lambda)Y) \leq \lambda \rho_t(X) + (1-\lambda)\rho_t(Y)$ for every $\lambda \in \LpK{\infty}{t}{[0,1]}$;
\item \emph{Finite at zero}: $\rho_t(\zero) \in \LiF{t}$.
\end{itemize}
A scalarized conditional convex risk measure is called \emph{coherent} if it additionally satisfies:
\begin{itemize}
\item \emph{Conditional positive homogeneity}: $\rho_t(\lambda X) = \lambda \rho_t(X)$ for every $\lambda \in \LdiK{t}{\R_+}$.
\end{itemize}
A scalarized conditional convex risk measure is called \emph{$\K_t$-compatible} if it additionally satisfies:
\begin{itemize}
\item $\rho_t(X) = \essinf_{k \in \K_t} \rho_t(X-k)$ for every $X \in \LdiF{}$, where $\K_t = \sum_{s = t}^T \LdpK{\infty}{s}{K_s}$ for $\Ft{s}$-measurable random convex cones $K_s$.
\end{itemize}
\end{definition}
$\K_t$-compatibility is used to allow the possibility of trading at time points $s \in\{t,...,T\}$ according to, e.g., the bid-ask prices modeled by a sequence of solvency cones $\seq{K}$.

The following two results relate the scalarized conditional risk measures $\rho_t$ with set-valued risk measures $R_t$ over a single eligible asset (provided by $M_t$).  Briefly, a set-valued risk measure is a mapping $R_t$ from the space of contingent claims into sets of eligible portfolios, i.e., capital allocations, that make the initial claim acceptable to a risk manager or regulator.  These risk measures satisfy cash invariance ($R_t(X + m) = R_t(X) - m$ for $X \in \LdiF{}$ and $m \in M_t$), monotonicity ($R_t(X) \subseteq R_t(Y)$ if $Y - X \in \LdiK{}{\R^d_+}$), conditional convexity ($R_t(\lambda X + (1-\lambda)Y) \supseteq \lambda R_t(X) + (1-\lambda) R_t(Y)$ for $X,Y \in \LdiF{}$ and $\lambda \in \LpK{\infty}{t}{[0,1]}$), and finiteness at zeros.  For more details on the set-valued conditional risk measure we refer to \cite{FR12,FR12b,FR13-survey}.
\begin{theorem}\label{thm_scalar}
Let $R_t: \LdiF{} \to \mathcal{G}(M_t;M_{t,+}) := \lrcurly{D \subseteq M_t \; | \; D = \cl\co\lrparen{D + M_{t,+}}}$ be a set-valued conditional convex (coherent) risk measure that is $\K_t$-compatible (i.e., $R_t(X) = \bigcup_{k \in \K_t} R_t(X-k)$ for every $X \in \LdiF{}$), 
then the mapping $\rho_t$ defined by
\begin{equation}\label{eq_scalar}
\rho_t(X) := \essinf\{m \in \LiF{t} \; | \; m e_1 \in R_t(X)\} \quad \forall X \in \LdiF{}
\end{equation}
is a scalarized conditional convex (resp.\ coherent) $\K_t$-compatible risk measure and, conversely,
\begin{equation*}
R_t(X) = (\rho_t(X) + \LdiK{t}{\R_+})e_1 \quad \forall X \in \LdiF{}.
\end{equation*}
\end{theorem}
\begin{proof}
First we will show that $\rho_t(X) \in \LiF{t}$ for any $X \in \LdiF{}$.  By assumption there exists some $Z^-,Z^+ \in \LiF{}$ such that $Z^-e_1 \in X - \K_t$ and $Z^+e_1 \in X + \K_t$.  This implies $R_t(Z^- e_1) \subseteq R_t(X) \subseteq R_t(Z^+ e_1)$ by $\K_t$-compatibility. This implies
\begin{align*}
\rho_t(X) &\geq \rho_t(Z^+ e_1) \geq \rho_t(\|Z^+\|_\infty e_1) \geq -(\|\rho_t(\zero)\|_\infty+\|Z^+\|_\infty) > -\infty\\
\rho_t(X) &\leq \rho_t(Z^- e_1) \leq \rho_t(-\|Z^-\|_\infty e_1) \leq \|\rho_t(\zero)\|_\infty+\|Z^-\|_\infty < +\infty.
\end{align*}
The remainder of the properties are proven in Theorem 3.15 and Corollary 3.17 of \cite{FR13-survey}.

Second, consider the converse statement $R_t(X) = (\rho_t(X) + \LdiK{t}{\R_+})e_1$ for every $X \in \LdiF{}$.
If $ue_1 \in R_t(X)$ then $u \geq \rho_t(X)$ by the definition of the scalarization.
The other direction follows if $\rho_t(X)e_1 \in R_t(X)$.  To show this we use the result from \cite[Lemma 3.18]{FR13-survey} that
\[R_t(X) = \bigcap_{w \in \plus{M_{t,+}} \backslash \prp{M_t}} \lrcurly{u \in M_t \; | \; \rho_t^{M,w}(X) \leq \trans{w}u}\]
where $\rho_t^{M,w}$ denotes the scalarization $\rho_t^{M,w}(X) := \essinf_{m \in R_t(X)} \trans{w}m$.
Then we note that 
\begin{align*}
\trans{w}(\rho_t(X)e_1) &= w_1 \rho_t(X) = w_1 \essinf\lrcurly{u \in \LiF{t} \; | \; ue_1 \in R_t(X)} \\
&= \essinf\lrcurly{\trans{w} m \; | \; m \in M_t,\; m \in R_t(X)} = \rho_t^{M,w}(X),
\end{align*}
which implies the result.
\end{proof}

\begin{assumption}
\label{ass_friction}
For the remainder of this paper we will assume that the conditions of Theorem~\ref{thm_scalar} are satisfied and that $\rho_t$ is constructed as in \eqref{eq_scalar}. 
We will additionally assume that $\K_t$ is defined w.r.t.\ solvency cones $\seq{K}$ with $K_T \supseteq \bar{K}$ almost surely for some convex cone $\bar{K} \subseteq \R^d$ with $\interior\bar{K} \supseteq \R^d_+ \backslash\{\zero\}$.
\end{assumption}  
The acceptance set $A_t \subseteq \LdiF{}$ is defined via the set-valued risk measure $R_t$, i.e., $A_t = \{X \in \LdiF{} \; | \; \zero \in R_t(X)\}$.  From the acceptance set we can also define $\rho_t$ as
\begin{equation}\label{A-rep}
\rho_t(X) := \essinf\{m \in \LiF{t} \; | \; X + m e_1 \in A_t\} \quad \forall X \in \LdiF{}.
\end{equation}  
Under this setting $\rho_t$ defines the minimal capital (in the first asset, e.g., Euros) necessary to hedge the risk of $X$ when trades defined by a market model $\K_t$ are allowed.

\begin{proposition}\label{prop_lipschitz}
Let $\rho_t: \LdiF{} \to \LiF{t}$ be a scalarized conditional convex risk measure.  Then it also satisfies the following two properties:
\begin{enumerate}
\item \emph{$\Ft{t}$-Lipschitz continuity}: $|\rho_t(X) - \rho_t(Y)| \leq \|X - Y\|_{\K_t,t}$ almost surely for every $X,Y \in \LdiF{}$, and
\item \emph{Local property}: $1_B \rho_t(X) = 1_B \rho_t(1_B X)$ for every $X \in \LdiF{}$ and $B \in \Ft{t}$
\end{enumerate}
where, as defined in \cite{TL12},
\begin{equation}\label{eq_norm}
\|X\|_{\K_t,s} := \essinf\{c \in \LiF{s} \; | \; ce_1 \geq_{\K_t} X \geq_{\K_t} -ce_1\}
\end{equation}
for $0\leq s\leq t\leq T$ and with $Y \geq_{\K_t} X$ if $Y - X \in \K_t$.
\end{proposition}
\begin{proof}
Let $X,Y \in \LdiF{}$ and $B \in \Ft{t}$.
\begin{enumerate}
\item By definition $\|X - Y\|_{\K_t,t}e_1 \geq_{\K_t} X - Y$.  By $\K_t$-compatibility this implies \[\rho_t(X) \geq \rho_t(Y + \|X - Y\|_{\K_t,t}e_1) = \rho_t(Y) - \|X - Y\|_{\K_t,t}.\]  This provides the desired result by symmetry of $X$ and $Y$.
\item By conditional convexity we recover 
\[\rho_t(1_B X) = \rho_t(1_B X + 1_{B^c} \zero) \leq 1_B \rho_t(X) + 1_{B^c} \rho_t(\zero),\] 
which implies $1_B \rho_t(1_B X) \leq 1_B \rho_t(X)$.  To prove the converse
\[\rho_t(X) = \rho_t(1_B [1_B X] + 1_{B^c} X) \leq 1_B \rho_t(1_B X) + 1_{B^c} \rho_t(X),\]
which implies $1_B \rho_t(X) \leq 1_B \rho_t(1_B X)$.
\end{enumerate}
\end{proof}

\subsection{Dual representation}
\label{sec_e1}

We will now consider the dual, or robust, representation for the scalar conditional risk measures. We present two equivalent formulations, both reliant on a lower semicontinuity property encoded in the following Fatou property.  The first of these two dual representations is intimately related to the dual representation of the underlying set-valued risk measures as discussed in Theorem~\ref{thm_scalar}.  See~\cite{FR12,FR12b,FR15-supermtg} for discussion on the dual representation for set-valued risk measures.  The second of these dual representations is a generalization of representation \eqref{JK d assets} of the multivariate superhedging price derived in \cite{JK95,LR11}.

\begin{definition}\label{defn_fatou}
Let $\rho_t$ be a scalarized conditional convex risk measure.  $\rho_t$ is said to satisfy the \textbf{\emph{Fatou property}} if
\[\rho_t(X) \leq \liminf_{n \to \infty} \rho_t(X_n)\]
for any $\|\cdot\|_\infty$-bounded sequence $X_n$ converging to $X$ almost surely.
\end{definition}

We will now provide the first dual representation for these scalarized conditional convex risk measures. This representation is similar to the results in \cite[Appendix A.2]{FR15-supermtg}.  Recall the set of dual variables $\W_t$ defined in~\eqref{eq_dualvar}.  Additionally, we will make extensive use of the notation introduced in Section~\ref{sec_prelim} that $w_t^s(\Q,w) = \diag{w}\xi_{t,s}(\Q)$ for $(\Q,w) \in \W_t$ with $\xi_{t,s}(\Q)$ the vector of conditional Radon-Nikodym derivatives of $\Q$ w.r.t.\ $\P$.

\begin{proposition}
\label{prop_scalar}
Let $\rho_t$ be a scalarized conditional convex risk measure. Then $\rho_t$ satisfies the Fatou property
if and only if for any $X \in \LdiF{}$ 
\begin{equation}
\label{eq_scalar_dynamic_dual}
\rho_t(X) = \esssup_{(\Q,m_{\perp}) \in \W_t(e_1)} \lrparen{-\alpha_t(\Q,m_{\perp}) - \transp{e_1 + m_{\perp}}\EQt{X}{t}},
\end{equation}
where 
\[\W_t(e_1) := \lrcurly{(\Q,m_{\perp}) \in \mathcal{M}^d \times \lrparen{\{0\} \times \LdoK{t}{\R^{d-1}_+}} \; | \; (\Q,e_1 + m_{\perp}) \in \W_t, \; w_t^T(\Q,e_1+m_{\perp}) \in \plus{\K_t}}\] 
and
\[\alpha_t(\Q,m_{\perp}) := \esssup_{Z \in A_t} \transp{e_1 + m_{\perp}}\EQt{-Z}{t}.\]
If $\rho_t$ is additionally a conditional coherent risk measure then \eqref{eq_scalar_dynamic_dual} can be reduced to
\begin{equation}
\label{eq_scalar_dynamic_dual_coherent}
\rho_t(X) = \esssup_{(\Q,m_{\perp}) \in \W_t^\rho(e_1)} -\transp{e_1 + m_{\perp}}\EQt{X}{t},
\end{equation}
for every $X \in \LdiF{}$ where $\W_t^\rho(e_1) := \lrcurly{(\Q,m_{\perp}) \; | \; \alpha_t(\Q,m_{\perp}) = 0 \; \P\text{-}\as}$.
\end{proposition}
\begin{proof}
\begin{enumerate}
\item Assume the dual representation~\eqref{eq_scalar_dynamic_dual} holds.  Let $X \in \LdiF{}$ and $(X_n)_{n \in \N} \subseteq \LdiF{}$ be a bounded sequence such that $X = \lim_{n \to \infty} X_n$ almost surely. Therefore, $\trans{Y}X_n \to \trans{Y}X$ a.s.\ and in $\LoF{}$ for any $Y \in \LdoF{}$ (by the dominated convergence theorem).  Thus $\Et{-\trans{w_t^T(\Q,e_1+m_{\perp})}X_n}{t} \to \Et{-\trans{w_t^T(\Q,e_1+m_{\perp})}X}{t}$ a.s.\ and in $\LoF{}$ for any $(\Q,m_{\perp}) \in \W_t(e_1)$ as well.
    Thus we can see (as in the static setting, e.g., in~\cite{FS11})
    \begin{align*}
    \rho_t(X) &= \esssup_{(\Q,m_{\perp}) \in \hatW_t(e_1)} \lrparen{-\alpha_t(\Q,m_{\perp}) + \transp{e_1+m_{\perp}}\EQt{-X}{t}}\\
    &= \esssup_{(\Q,m_{\perp}) \in \hatW_t(e_1)} \lrparen{-\alpha_t(\Q,m_{\perp}) + \lim_{n \to \infty} \transp{e_1+m_{\perp}}\EQt{-X_n}{t}}\\
    &= \esssup_{(\Q,m_{\perp}) \in \hatW_t(e_1)} \lim_{n \to \infty} \lrparen{-\alpha_t(\Q,m_{\perp}) + \transp{e_1+m_{\perp}}\EQt{-X_n}{t}}\\
    &\leq \liminf_{n \to \infty} \esssup_{(\Q,m_{\perp}) \in \hatW_t(e_1)} \lrparen{-\alpha_t(\Q,m_{\perp}) + \transp{e_1+m_{\perp}}\EQt{-X_n}{t}}\\
    &= \liminf_{n \to \infty} \rho_t(X_n).
    \end{align*}
\item Assume that the ``Fatou property'' is satisfied.  To prove the dual representation \eqref{eq_scalar_dynamic_dual} holds, we only need to prove that $\E{\rho_t(\cdot)}$ is proper and lower semicontinuous from the same logic as in Proposition A.1.7 of \cite{FR15-supermtg}.  This modification to the conditions in \cite[Proposition A.1.7]{FR15-supermtg} can be accomplished because the requirements on the set-valued risk measures in that result are utilized, via \cite[Proposition A.1.1]{FR15-supermtg} and its proof, to guarantee that $\E{\rho_t(\cdot)}$ is proper and lower semicontinuous. By $\rho_t(X) \in \LiF{t}$ for every $X \in \LdiF{}$, we can conclude $\E{\rho_t(\cdot)}$ is proper.
    Define $C_z := \lrcurly{Z \in \LdiF{} \; | \; \E{\rho_t(Z)} \leq z}$ for any $z \in \R$.  Define $C_z^r := C_z \cap \lrcurly{Z \in \LdiF{} \; | \; \|Z\|_\infty \leq r}$ for any $r > 0$ (in this case we will take $\|Z\|_\infty := \esssup \max_{i = 1,...,d} \abs{Z_i}$).  Take $(X_n)_{n \in \N} \subseteq C_z^r$ converging to $X \in \LdiF{}$ in probability, and there exists a bounded subsequence $(X_{n_m})_{m \in \N}$ which converges to $X$ almost surely.  Therefore $\rho_t(X) \leq \liminf_{m \to \infty} \rho_t(X_{n_m})$, and thus (by Fatou's lemma) $X \in C_z^r$.  That is $C_z^r$ is closed in probability for any $r > 0$.  By~\cite[Proposition 5.5.1]{KS09}, $C_z$ is weak* closed.  Now consider a net $(X_i)_{i \in I} \subseteq \LdiF{}$ which converges to $X \in \LdiF{}$ in the weak* topology.  Let $z_{\epsilon} := \E{\rho_t(X)} - \epsilon$ for any $\epsilon > 0$. Then $(C_{z_{\epsilon}})^c = \lrcurly{Z \in \LdiF{} \; | \; \E{\rho_t(Z)} > z_{\epsilon}}$ is an open neighborhood of $X$ for any $\epsilon > 0$.  Since $X_i \to X$, for any $\epsilon > 0$ there exists a $j_{\epsilon} \in I$ such that for any $i \geq j_{\epsilon}$ we have $X_i \in (C_{z_{\epsilon}})^c$, i.e., $\E{\rho_t(X_i)} > \E{\rho_t(X)} - \epsilon$.  Therefore $\liminf_{i \in I} \E{\rho_t(X_i)} \geq \E{\rho_t(X)}$, and the proof is complete.
\end{enumerate}
Finally, the coherent case follows identically to Corollary 2.5 of~\cite{FP06}.
\end{proof}

\begin{remark}\label{rem:set-fatou}
For the reader who is interested in the relation to set-valued risk measures, let us note that Proposition A.1.3 of \cite{FR15-supermtg} provides a sufficient condition on the underlying set-valued risk measure $R_t$ (with any choice of eligible space $\tilde{M} \supseteq M$) to guarantee that $\rho_t$ satisfies the Fatou property. 
\end{remark}

With the dual representation from Proposition~\ref{prop_scalar}, we construct a second dual representation with a single probability measure and with respect to vectors of consistent prices.  This is a generalization of the dual representation of the superhedging risk measure from \cite{JK95}. This new dual representation is valuable for two primary reasons: it provides a clear interpretation as the supremum over all frictionless pricing processes consistent with the market model and, as a consequence, it allows for a decomposition of the risk measure into frictionless markets that allows for our new time consistency notion discussed in Section~\ref{sec_tc}.  For this representation, recall the definition of $\|\cdot\|_{\K_T,0}$ from~\eqref{eq_norm}.

\begin{theorem}
\label{thm_dual_e1}
Let $\rho_t$ be a scalarized conditional convex risk measure.  Then $\rho_t$ satisfies the Fatou property 
if and only if for any $X \in \LdiF{}$ 
\begin{equation}\label{eq_dual_e1}
\rho_t(X) = \esssup_{(\Q,S) \in \QQ_t} \lrparen{-\beta_t(\Q,S) - \EQt{\trans{S_T}X}{t}}
\end{equation}
where \[\beta_t(\Q,S) = \esssup_{Z \in A_t} -\EQt{\trans{S_T}Z}{t}\] and
\[\QQ_t = \lrcurly{(\Q,(S_s)_{s = t}^T) \in \mathcal{M} \times \prod_{s = t}^T \LdiF{s} \; \left| \; \begin{array}{l} \forall s = t,...,T,\; \forall i = 1,...,d:\\ S_{s,1} \equiv 1, \; \|S_{s,i}\|_\infty \leq \max\{\|e_i\|_{\K_T,0} , 1\}, \\ S_s = \EQt{S_T}{s}, \; \dQdP S_T \in \plus{\K_T} \end{array} \right.}.\]
\end{theorem}
\begin{proof}
By Proposition~\ref{prop_scalar} the result trivially follows if
\begin{enumerate}
\item for every $(\R,m_{\perp}) \in \hatW_t(e_1)$ with $\E{\alpha_t(\R,m_{\perp})} < +\infty$ there exists a $(\Q,S) \in \QQ_t$ such that $w_t^s(\R,e_1+m_{\perp}) = \bar\xi_{t,s}(\Q) S_s$ for every time $s \geq t$, and
\item for every $(\Q,S) \in \QQ_t$ there exists a $(\R,m_{\perp}) \in \hatW_t(e_1)$ such that $w_t^s(\R,e_1+m_{\perp}) = \bar\xi_{t,s}(\Q) S_s$ for every time $s$. 
\end{enumerate}

Let us show that the first property holds. Let $(\R,m_{\perp}) \in \hatW_t(e_1)$ with $\E{\alpha_t(\R,m_{\perp})} < +\infty$.  Let $Z_s := w_t^s(\R,e_1+m_{\perp})$; note that $(Z_s)_{s = t}^T$ is a $\P$-martingale.  Then we will define $S_s$ by $S_{s,i} = \begin{cases} \frac{Z_{s,i}}{Z_{s,1}} &\text{on } \{Z_{s,1} > 0\}\\ 1 &\text{else} \end{cases}$ for all indices $i = 1,...,d$.  Note that $S_{s,1} \equiv 1$ for all times $s$ and $S_t = e_1 + m_{\perp} \in \LdiK{t}{\R^d_+}$.  Further, define $\Q \in \mathcal{M}$ by $\dQdP = \frac{Z_{T,1}}{Z_{t,1}} = Z_{T,1}$ (since $Z_t = e_1+m_{\perp}$).  Notice that $\{Z_{s,1} = 0\} \subseteq \{Z_{T,1} = 0\}$ and by property of $\K_t$ and Assumption~\ref{ass_friction} we have that $\{Z_{T,1} = 0\} = \{Z_{T,i} = 0\}$ (as $Z_T \in \K_t^+$) for any choice of index $i$, therefore we can see that
\begin{align*}
\EQt{S_{T,i}}{s} &= \Et{1_{\{\Et{\dQdP}{s} > 0\}} \frac{\dQdP}{\Et{\dQdP}{s}}S_{T,i} + 1_{\{\Et{\dQdP}{s} = 0\}} S_{T,i}}{s}\\
&= \Et{\begin{array}{l} 1_{\{Z_{s,1} > 0\}} \frac{Z_{T,1}}{Z_{s,1}} \lrparen{1_{\{Z_{T,1} > 0\}} \frac{Z_{T,i}}{Z_{T,1}} + 1_{\{Z_{T,1} = 0\}}}\\ + 1_{\{Z_{s,1} = 0\}} \lrparen{1_{\{Z_{T,1} > 0\}} \frac{Z_{T,i}}{Z_{T,1}} + 1_{\{Z_{T,1} = 0\}}}\end{array}}{s}\\
&= \Et{\begin{array}{l} 1_{\{Z_{s,1} > 0\}}1_{\{Z_{T,1} > 0\}} \frac{Z_{T,i}}{Z_{s,1}} + 1_{\{Z_{s,1} > 0\}}1_{\{Z_{T,1} = 0\}} \frac{Z_{T,1}}{Z_{s,1}}\\ + 1_{\{Z_{s,1} = 0\}}1_{\{Z_{T,1} > 0\}} \frac{Z_{T,i}}{Z_{T,1}} + 1_{\{Z_{s,1} = 0\}}1_{\{Z_{T,1} = 0\}}\end{array}}{s}\\
&= 1_{\{Z_{s,1} > 0\}} \frac{1}{Z_{s,1}} \Et{1_{\{Z_{T,1} > 0\}} Z_{T,i}}{s} + 1_{\{Z_{s,1} = 0\}}\\
&= 1_{\{Z_{s,1} > 0\}} \frac{1}{Z_{s,1}} \Et{1_{\{Z_{T,i} > 0\}} Z_{T,i}}{s} + 1_{\{Z_{s,1} = 0\}}\\
&= 1_{\{Z_{s,1} > 0\}} \frac{1}{Z_{s,1}} \Et{Z_{T,i}}{s} + 1_{\{Z_{s,1} = 0\}}\\
&= 1_{\{Z_{s,1} > 0\}} \frac{Z_{s,i}}{Z_{s,1}} + 1_{\{Z_{s,1} = 0\}} = S_{s,i}
\end{align*}
for any $s \in \{t,...,T\}$ and any $i = 1,...,d$.  Finally 
\begin{align*}
\bar\xi_{t,s}(\Q) S_s &= Z_{s,1} \lrparen{1_{\{Z_{s,1} > 0\}} \frac{Z_s}{Z_{s,1}} + 1_{\{Z_{s,1} = 0\}}} = 1_{\{Z_{s,1} > 0\}} Z_s\\
&= 1_{\{\bar\xi_{t,s}(\R_1) > 0\}} w_t^s(\R,e_1+m_{\perp}) \in \plus{\K_t}.
\end{align*}
It remains to show that $S_T \in \LdiK{}{\R^d_+}$ (the case for $s < T$ would then follow trivially): by definition $\|e_i\|_{\K_T,0} e_1 \geq_{\K_T} e_i \geq_{\K_T} -\|e_i\|_{\K_T,0} e_1$ for any index $i$ and, by the definition of $\K_t$, we know $\|e_i\|_{\K_T,0} < \infty$.  This implies $\dQdP \|e_i\|_{\K_T,0} \geq \dQdP S_{T,i} \geq -\dQdP \|e_i\|_{\K_T,0}$ (noting that $\plus{\K_t} \subseteq \plus{\K_T} = \lrcurly{Z \in \LdoF{} \; | \; \trans{Z}K \geq 0 \; \P\text{-}\as \; \forall K \in \LdiK{}{K_T}}$), by dividing by $\dQdP$ on the set $\{\dQdP > 0\}$ recovers $S_{T,i} \in \Lnp{}{\infty}{\Omega,\Ft{s},\Q;\R}$ with $\|S_{T,i}\|_\infty^\Q \leq \|e_i\|_{\K_T,0}$ (i.e., the $\infty$-norm of $S_{T,i}$ under measure $\Q$ is bounded).  Finally, by construction $S_{T,i} = 1$ on $\{\dQdP = 0\}$, thus $\|S_{T,i}\|_\infty \leq \max\{\|e_i\|_{\K_T,0} , 1\}$.  Since $S$ is a $\Q$-martingale, the bound $\|S_{s,i}\|_\infty \leq \max\{\|e_i\|_{\K_T,0} , 1\}$ holds as well.

To show the converse, i.e., the second property, let $(\Q,S) \in \QQ_t$.  Define $m_{\perp} = S_t - e_1 \in \{0\} \times \LdpK{1}{t}{\R^{d-1}}$.  And define $\dRidP = \begin{cases} \dQdP \frac{S_{T,i}}{S_{t,i}} &\text{on }
\lrcurly{S_{t,i} > 0} \\ 1 &\text{else} \end{cases}$.  Then $w_t^s(\R,e_1+m_{\perp})_i = S_{t,i} \Et{\dRidP}{s} = S_{t,i} \Et{\dQdP \frac{S_{T,i}}{S_{t,i}} 1_{\{S_{t,i} > 0\}} + 1_{\{S_{t,i} = 0\}}}{s} = \Et{\dQdP S_{T,i}}{s} = \bar\xi_{t,s}(\Q) \EQt{S_{T,i}}{s} = \bar\xi_{t,s}(\Q) S_{s,i}$ for every index $i = 1,...,d$ (because $\lrcurly{S_{t,i} = 0} \subseteq \lrcurly{\dQdP S_{T,i} = 0}$ by $S$ a $\Q$-martingale). 
\end{proof}

\begin{remark}\label{rem_dualvar}
For every choice of dual variables $(\Q,S) \in \QQ_t$ we find that the dual norm $\|\dQdP S_T\|_{\K_T,t}^* := \esssup\{|\EQt{\trans{S_T}X}{t}| \; | \; \|X\|_{\K_T,t} \leq 1 \; \P\text{-}\as\} = 1$ almost surely.  
\end{remark}

\begin{corollary}
\label{cor_dual_e1_coherent}
Let $\rho_t$ be a scalarized conditional coherent risk measure.  Then $\rho_t$ satisfies the Fatou property 
if and only if for any $X \in \LdiF{}$ 
\[\rho_t(X) = \esssup_{(\Q,S) \in \QQ_t^{\rho}} -\EQt{\trans{S_T}X}{t}\]
where $\QQ_t^{\rho} = \lrcurly{(\Q,S) \in \QQ_t \; | \; \beta_t(\Q,S) = 0 \; \P\text{-}\as}$.
\end{corollary}
\begin{proof}
This follows directly from Theorem~\ref{thm_dual_e1} and Proposition~\ref{prop_scalar}. 
\end{proof}

\begin{remark}\label{rem_SHP_dual}
The $d$ dimensional version of the dual representation of the scalar superhedging price~\eqref{JK d assets} given in Jouini, Kallal \cite{JK95} is with respect to dual variables $(\Q,S) \in \mathcal{M} \times \LdpK{0}{}{\R^d_+}$ with $S_s \equiv 1$ for every $s$, $\Q$ is a martingale measure for $S$, and $\dQdP S_T \in \plus{\K_0}$.  This set of dual variables is a superset of $\QQ_t$ from Theorem~\ref{thm_dual_e1}. From the proof of Theorem~\ref{thm_dual_e1}, it is clear that this larger set of dual variables could be used within this work as well.  However, for reasons that will be clear in subsequent results, e.g., Proposition~\ref{prop:pi_inf}, which is used extensively after, we restrict ourselves to the bounded prices as defined in $\QQ_t$. \end{remark}

\subsection{Relevance}
\label{sec_relevance}

In the below proposition we will define a property so that the dual representation~\eqref{eq_dual_e1} can be defined w.r.t.\ the dual variables 
\[\QQ_t^e := \{(\Q,S) \in \QQ_t \; | \; \Q \in \mathcal{M}^e\},\] 
i.e., the set of equivalent martingale measures and their associated price processes $S$. We will then demonstrate that this necessary and sufficient condition is satisfied under a version of relevance or sensitivity (see the usual definition in, e.g., \cite{FP06}).

\begin{remark}\label{rem:QQ}
Before continuing, note that $(\Q,S) \in \QQ_t^e$ if and only if $\Q \in \mathcal{M}^e$ is a martingale measure for $S$ and $S \in \prod_{s = t}^T \LdiK{s}{\plus{K_s}}$ such that $S_{s,1} \equiv 1$ and $\|S_{s,i}\|_\infty \leq \max\{\|e_i\|_{\K_T,0} , 1\}$.  This is in contrast to the definition of $\QQ_t$ in which $\dQdP S_T \in \plus{\K_t}$.  This modification can be accomplished as we are now only using equivalent measures $\Q$ and since $\plus{\K_t} = \bigcap_{s = t}^T \lrcurly{Y \in \LdoF{} \; | \; \trans{\Et{Y}{s}}k_s \geq 0 \; \P\text{-}\as \; \forall k_s \in \LdiK{s}{K_s}}$. 
\end{remark}

\begin{proposition}
\label{prop_equivalent-prob}
Let $\rho_t$ be a scalarized conditional convex risk measure satisfying the Fatou property. Then
\begin{equation}
\label{eq_equivalent-prob}
\rho_t(X) = \esssup_{(\Q,S) \in \QQ_t^e} \lrparen{-\beta_t(\Q,S) - \EQt{\trans{S_T}X}{t}}
\end{equation}
for every $X \in \LdiF{}$ if and only if $\inf_{Z \in A_t} \EQ{\trans{S_T} Z} > -\infty$ for some $(\Q,S) \in \QQ_t^e$.
\end{proposition}
\begin{proof}
If $\rho_t(X) = \esssup_{(\Q,S) \in \QQ_t^e} \lrparen{-\beta_t(\Q,S) - \EQt{\trans{S_T}X}{t}}$ for every $X \in \LdiF{}$ then, by $\Ft{t}$-decomposability of $\{\dQdP\trans{S_T}Z \; | \; (\Q,S) \in \QQ_t^e, \; Z \in A_t\}$ (recall from Section~\ref{sec_prelim} that a set $B$ is called $\Ft{t}$-decomposable if $1_D B + 1_{D^c} B \subseteq B$ for any $D \in \Ft{t}$) and the definition of the penalty function $\beta_t$, $\E{\rho_t(\zero)} = \sup_{(\Q,S) \in \QQ_t^e} \inf_{Z \in A_t} \E{\dQdP\trans{S_T} Z} > -\infty$ by $\rho_t(\zero) \in \LiF{t}$.  In particular this implies that there exists some $(\Q,S) \in \QQ_t^e$ such that $\inf_{Z \in A_t} \EQ{\trans{S_T} Z} > -\infty$.

Conversely, let $(\tilde\Q,\tilde{S}) \in \QQ_t^e$ such that $\inf_{Z \in A_t} \EtQ{\trans{\tilde{S}_T} Z} > -\infty$.  Since $\QQ_t^e \subseteq \QQ_t$, the dual representation in Theorem~\ref{thm_dual_e1} implies that for every $X \in \LdiF{}$, $\rho_t(X) \geq \esssup_{(\Q,S) \in \QQ_t^e} \lrparen{-\beta_t(\Q,S) - \EQt{\trans{S_T}X}{t}}$.
As in the proof of Proposition A.1.7 of \cite{FR15-supermtg} we will show the reverse inequality by considering the expectation.
Since $\{-\beta_t(\Q,S) - \EQt{\trans{S_T}X}{t} \; | \; (\Q,S) \in \QQ_t\}$ is $\Ft{t}$-decomposable we are able to interchange the expectation and essential supremum.  
As done in~\cite[Lemma 3.5]{FP06}, for any $\epsilon \in (0,1)$ and $(\Q,S) \in \QQ_t$, define
$(\Q_\epsilon,S_\epsilon) := (1-\epsilon) (\Q,S) + \epsilon (\tilde\Q,\tilde{S}) \in \QQ_t^e$.
Additionally, by definition of the penalty function, we can conclude
$\beta_t(\Q_{\epsilon},S_{\epsilon}) \leq (1-\epsilon)\beta_t(\Q,S) + \epsilon \beta_t(\tilde\Q,\tilde{S})$.
Choose any $X \in \LdiF{}$
\begin{align*}
\E{\rho_t(X)} &= \sup_{(\Q,S) \in \QQ_t} \E{-\beta_t(\Q,S) - \dQdP \trans{S_T} X}\\
&\geq \sup_{(\Q,S) \in \QQ_t^e} \E{-\beta_t(\Q,S) - \dQdP \trans{S_T}X}\\
&\geq \sup_{(\Q,S) \in \QQ_t} \E{-\beta_t(\Q_{\epsilon},S_{\epsilon}) - \frac{d\Q_\epsilon}{d\P}\trans{S_{\epsilon,T}}X}\\
&\geq \sup_{(\Q,S) \in \QQ_t} \lparen{(1-\epsilon) \E{-\beta_t(\Q,S) - \dQdP \trans{S_T}X}}\\
&\quad\quad \rparen{+ \epsilon \E{-\beta_t(\tilde\Q,\tilde{S}) - \frac{d\tilde\Q}{d\P}\trans{\tilde{S}_T}X}}\\
&= (1-\epsilon) \E{\rho_t(X)} + \epsilon \E{-\beta_t(\tilde\Q,\tilde{S}) - \frac{d\tilde\Q}{d\P}\trans{\tilde{S}_T}X}.
\end{align*}
Since $\E{\rho_t(X)} \in \R$ then we take the limit as $\epsilon$ tends to $0$ to recover that the final line above is equivalent to the first line, and the result is shown.
\end{proof}

\begin{corollary}
\label{cor_equivalent-prob}
Let $\rho_t$ be a scalarized conditional coherent risk measure satisfying the Fatou property. Then
\begin{equation}
\label{eq_coherent_equivalent-prob}
\rho_t(X) = \esssup_{(\Q,S) \in \QQ_t^{\rho,e}} -\EQt{\trans{S_T}X}{t}
\end{equation}
for every $X \in \LdiF{}$ if and only if $\QQ_t^{\rho,e} := \QQ_t^e \cap \QQ_t^\rho \neq \emptyset$, i.e., there exists some $(\Q,S) \in \QQ_t^e$ such that $\beta_t(\Q,S) = 0$ almost surely.
\end{corollary}
\begin{proof}
Recall from the logic of the proof of Corollary~2.5 of~\cite{FP06}, for any $(\Q,S) \in \QQ_t$ we have that $\E{\beta_t(\Q,S)} < +\infty$ if and only if $(\Q,S) \in \QQ_t^{\rho}$.  The results follows from Proposition~\ref{prop_equivalent-prob} above.  If $\rho_t(X) = \esssup_{(\Q,S) \in \QQ_t^{\rho,e}} -\EQt{\trans{S_T}X}{t}$ then there exists some $(\Q,S) \in \QQ_t^e$ such that $\inf_{Z \in A_t} \EQ{\trans{S_T} Z} > -\infty$ (else $\E{\rho_t(\cdot)} \equiv -\infty$).  Conversely, if $\inf_{Z \in A_t} \E{\trans{S_T} Z} > -\infty$ for some $(\Q,S) \in \QQ_t^e$ then we have
\begin{align*}
\rho_t(X) &= \esssup_{\substack{(\Q,S) \in \QQ_t^e \\ \E{\beta_t(\Q,S)} < +\infty}} \lrparen{-\beta_t(\Q,S) - \EQt{\trans{S_T}X}{t}}\\
&= \esssup_{(\Q,S) \in \QQ_t^{\rho,e}} -\EQt{\trans{S_T}X}{t}.
\end{align*}
\end{proof}

As in the literature for scalar risk measures, we will now introduce relevance (or sensitivity) and demonstrate that it implies the
existence of $(\Q,S) \in \QQ_t^e$ such that $\inf_{Z \in A_t} \EQ{\trans{S_T} Z} >
-\infty$.  See, e.g., \cite{FP06,CDK06,KS07} for previous literature on relevance of univariate conditional risk measures. 
\begin{definition}
\label{defn_relevance}
A risk measure $\rho_t$ is called \textbf{\emph{relevant}} if 
\[\P(\rho_t(-\epsilon 1_D e_1) > \rho_t(\zero)) > 0\]
for every $\epsilon > 0$ and every $D \in \Ft{}$ with $\P(D) > 0$. 
\end{definition}

\begin{lemma}
\label{lemma_relevance}
Let $\rho_t$ be a scalarized conditional convex risk measure satisfying the Fatou property. If $\rho_t$ is relevant then there exists $(\Q,S) \in \QQ_t^e$ such that $\inf_{Z \in A_t} \EQ{\trans{S_T} Z} > -\infty$.
\end{lemma}
\begin{proof}
First, $\lrcurly{\beta_t(\Q,S) \; | \; (\Q,S) \in \QQ_t}$ is $\Ft{t}$-decomposable (by the definition of $\QQ_t$ and $\beta_t$).  Therefore $\R \ni \E{-\rho_t(\zero)} = \E{\essinf_{(\Q,S) \in \QQ_t} \beta_t(\Q,S)} = \inf_{(\Q,S) \in \QQ_t} \E{\beta_t(\Q,S)}$. 
For any $\epsilon > 0$ let the set $\QQ_t^\epsilon \subseteq \QQ_t$ be defined by
\[\QQ_t^{\epsilon} := \lrcurly{(\Q,S) \in \QQ_t \; | \; \beta_t(\Q,S) < -\rho_t(\zero) + \epsilon \; \P\text{-}\as}.\]

Now we will show that for any $\epsilon > 0$ there exists a pair of dual variables $(\tilde\Q,\tilde{S}) \in \QQ_t^e \cap \QQ_t^\epsilon$ (and in particular this implies $\inf_{Z \in A_t} \EtQ{\trans{\tilde{S}_T} Z} > -\infty$).
\begin{enumerate}
\item We will show that $\QQ_t^{\epsilon} \neq \emptyset$ for any $\epsilon > 0$.  Let $(\Q_n,S_n)_{n \in \N} \subseteq \QQ_t$ such that $\beta_t(\Q_n,S_n) \searrow -\rho_t(\zero)$ almost surely.  This sequence exists because $\lrcurly{\beta_t(\Q,S) \; | \; (\Q,S) \in \QQ_t}$ is $\Ft{t}$-decomposable, thus there exists a sequence $(\Q_n,S_n)_{n \in \N} \subseteq \QQ_t$ such that \[\beta_t(\Q_n,S_n) \searrow \essinf_{(\Q,S) \in \QQ_t} \beta_t(\Q,S) = -\rho_t(\zero) \quad \P\text{-}\as\]
    Define the $\Ft{t}$-measurable random variable $\tau_{\epsilon}$ by
    \[\tau_{\epsilon} := \min\lrcurly{n \; | \; \beta_t(\Q_n,S_n) < -\rho_t(\zero) + \epsilon}.\]
    It holds $\tau_{\epsilon} < +\infty$ almost surely by construction of the sequence $(\Q_n,S_n)_{n \in \N}$, and thus $(\lrcurly{\tau_{\epsilon} = n})_{n \in \N} \subseteq \Ft{t}$ defines a partition of $\Omega$. We can define $(\bar\Q_\epsilon,\bar{S}_\epsilon) := \sum_{n = 1}^{\infty} 1_{\lrcurly{\tau_{\epsilon} = n}} (\Q_n,S_n) \in \QQ_t$ (trivially) with $\beta_t(\bar \Q_{\epsilon},\bar{S}_\epsilon) \leq \sum_{n = 1}^{\infty} 1_{\lrcurly{\tau_{\epsilon} = n}} \beta_t(\Q_n,S_n) < -\rho_t(\zero) + \epsilon$, i.e., $(\bar \Q_{\epsilon},\bar{S}_\epsilon) \in \QQ_t^{\epsilon}$.

\item We will show that if $\rho_t$ is relevant then for every $\epsilon > 0$ there exists an element $(\tilde\Q,\tilde{S}) \in \QQ_t^{\epsilon}$ such that $\frac{d\tilde\Q}{d\P} > 0$ a.s.  Fix $\epsilon > 0$ and let $c_\epsilon := \sup \lrcurly{\P(\dQdP > 0) \; | \; (\Q,S) \in \QQ_t^{\epsilon}}$.
Take a sequence $(\Q_n,S_n)_{n \in \N} \subseteq \QQ_t^{\epsilon}$ such that $\lim_{n \to \infty} \P(\dQndP{n} > 0) = c_{\epsilon}$.  Define $(\tilde\Q,\tilde{S}) := \sum_{n = 1}^{\infty} \frac{1}{2^n} (\Q_n,S_n)$. It holds $(\tilde\Q,\tilde{S}) \in \QQ_t^\epsilon$ by convexity of $\beta_t$. In fact, $\lrcurly{\frac{d\tilde\Q}{d\P} > 0} = \cup_{n = 1}^{\infty} \lrcurly{\dQndP{n} > 0}$.  Therefore $\P(\frac{d\tilde\Q}{d\P} > 0) = c_\epsilon$.  It remains to show that $c_\epsilon = 1$, and thus $\tilde\Q \in \mathcal{M}^e$.

    Assume $c_\epsilon < 1$, and let $D := \lrcurly{\frac{d\tilde\Q}{d\P} = 0}$.  Relevance implies
\[\P(\rho_t(-\epsilon 1_D e_1) > \rho_t(\zero)) > 0,\]
and by construction we have
\[\rho_t(-\epsilon 1_D e_1) = \esssup_{(\Q,S) \in \QQ_t} \lrparen{-\beta_t(\Q,S) + \epsilon \EQt{1_D}{t}}.\]
Therefore, there exists $(\bar{\Q},\bar{S}) \in \QQ_t$ such that the set
\[B := \lrcurly{\beta_t(\bar{\Q},\bar{S}) < -\rho_t(\zero) + \epsilon \EbQt{1_D}{t}} \in \Ft{t}\]
has positive probability, i.e., $\beta_t(\bar{\Q},\bar{S}) < -\rho_t(\zero) + \epsilon$ on $B$.

    Now we will define $(\hat\Q,\hat{S}) \in \QQ_t^{\epsilon}$ which is equal to $(\bar{\Q},\bar{S})$ on $B$.  Let $(\Q,S) \in \QQ_t^{\epsilon}$ (arbitrary) and let $(\hat\Q,\hat{S}) := 1_B (\bar{\Q},\bar{S}) + 1_{B^c} (\Q,S) \in \QQ_t^{\epsilon}$ (because $\beta_t(\hat\Q,\hat{S}) = 1_B \beta_t(\bar{\Q},\bar{S}) + 1_{B^c} \beta_t(\Q,S)$).

    By definition of $B$ and since $\beta_t(\bar{\Q},\bar{S}) \geq -\rho_t(\zero)$ (by the dual representation of $\rho_t(\zero)$),     
    \[\E{1_B 1_D \frac{d\hat\Q}{d\P}} = \E{1_B \Et{1_D \frac{d\hat\Q}{d\P}}{t}} > 0.\]
    This implies that $\P\lrparen{\lrcurly{\frac{d\hat\Q}{d\P} > 0} \cap D} \geq \P\lrparen{\lrcurly{\frac{d\hat\Q}{d\P} > 0} \cap D \cap B} > 0$.  Therefore $\frac{1}{2} (\tilde\Q,\tilde{S}) + \frac{1}{2} (\hat\Q,\hat{S}) \in \QQ_t^{\epsilon}$ by convexity of $\QQ_t^\epsilon$, and we have found a contradiction.
\end{enumerate}
\end{proof}

We now show that, in the coherent case, relevance is actually necessary and sufficient for a dual representation as in \eqref{eq_coherent_equivalent-prob}.

\begin{proposition}
\label{cor_relevance-coherent}
Let $\rho_t$ be a scalarized conditional coherent risk measure satisfying the Fatou property.  $\rho_t$ is relevant if and only if \eqref{eq_coherent_equivalent-prob} holds for every $X \in \LdiF{}$.
\end{proposition}
\begin{proof}
As shown above in Corollary~\ref{cor_equivalent-prob} and Lemma~\ref{lemma_relevance}, if $\rho_t$ is relevant then \eqref{eq_coherent_equivalent-prob} holds for every $X \in \LdiF{}$.  Conversely assume that $\rho_t$ has dual representation with respect to the dual variables $\QQ_t^{\rho,e}$ and assume $\rho_t$ is not relevant, i.e., there exists some $\epsilon > 0$ and $D \in \Ft{}$ with positive probability such that $\rho_t(-\epsilon 1_D e_1) \leq \rho_t(\zero) = 0$ almost surely.  
In particular this implies $\epsilon\EQt{1_D}{t} \leq \rho_t(-\epsilon 1_D e_1) \leq 0$ for any $(\Q,S) \in \QQ_t^{\rho,e}$.
Therefore $\epsilon \Q(D) \leq 0$ for every $(\Q,S) \in \QQ_t^{\rho,e}$.  However, this can only be true if $\Q(D) = 0$, which implies $\P(D) = 0$ since $\Q \in \mathcal{M}^e$.  This is a contradiction.
\end{proof}

\section{Time consistency}\label{sec_tc}

In this section we will study a modified version of the traditional time consistency property.  Typically, as studied in e.g., \cite{HMBS16}, the time consistency of scalarized risk measures is considered for $\rho$ directly, i.e.,
\begin{equation}\label{eq:rho-tc}
\rho_s(X) \leq \rho_s(Y) \Rightarrow \rho_t(X) \leq \rho_t(Y)
\end{equation}
for $X,Y \in \LdiF{}$ and $t \leq s$.  
As we will demonstrate in Section~\ref{sec_shp} the superhedging portfolios do not satisfy this \textbf{$\rho$-time consistency}.  Instead we will consider a decomposition of the scalarizations via their dual representations as defined in Theorem~\ref{thm_dual_e1}.  These details are provided in Section~\ref{sec_tc-prelim}.  Then in Section~\ref{sec:pi-tc}, we consider a new time consistency property.  In particular, with this new property we determine the typical equivalence of time consistency and a recursive relation.

\begin{assumption}\label{ass:fatou}
For the remainder of this paper we will assume that $\rho_t$ is a relevant scalarized conditional convex risk measure satisfying the Fatou property, and corresponds with the acceptance set $A_t$ as given in \eqref{A-rep}.
\end{assumption}

\subsection{Overview}\label{sec_tc-overview}
Before continuing to the main results of this work in regards to time consistency of the dynamic multivariate risk measures $\rho$, we will provide a brief overview over the topic, illustrate the importance of the derived results, set them in relation to existing results in the literature, and provide an economic interpretation.

As mentioned previously, the usual time consistency property, $\rho$-time consistency, is in general not satisfied for dynamic multivariate risk measures $\rho$ given by~\eqref{eq_scalar}. In particular, it does not hold for the prototypical case of the superhedging risk measure as detailed in, e.g., Section~\ref{sec_shp} below.  This is because of the frictions inherent to the risk measure in covering all risks with the cash asset alone, i.e., by considering a market with transaction costs.
In contrast, many prior works (e.g., \cite{FR12,FR12b}) on time consistency of multivariate risk measures have considered set-valued risk measures and a property called multiportfolio time consistency.  A set-valued risk measure is said to be \textbf{multiportfolio time consistent} if
\begin{equation}\label{eq_mptc}
\bigcup_{X \in \mathcal{X}} R_s(X) \supseteq R_s(Y) \; \Rightarrow \; \bigcup_{X \in \mathcal{X}} R_t(X) \supseteq R_t(Y)
\end{equation}
for all $\mathcal{X} \subseteq \LdiF{}$, $Y \in \LdiF{}$, and $t \leq s$.  This property relies on the choice of eligible assets $\tilde M$ taken in the image space for the set-valued risk measure $R_t: \LdiF{} \to \mathcal{G}(\tilde M_t;\tilde M_{t,+})$.  For instance, the set-valued superhedging risk measure is multiportfolio time consistent w.r.t.\ the full space of eligible assets (i.e., when one chooses $\tilde M = \R^d$) but not with the cash asset only (i.e., when $\tilde M = M=\R \times \{0\}^{d-1}$ as assumed within this paper), see Section~\ref{sec_shp} below for details.  As with $\rho$-time consistency, this is due to the frictions inherent to the risk measure when restricting to covering risk with only the cash asset.

However, in Section~\ref{sec_tc-prelim}, we will introduce a decomposition of the conditional risk measure $\rho_t$ into a collection of multivariate risk measures in frictionless markets, each of which is defined by a price process satisfying a no-arbitrage condition. This decomposition allows us to recover the desired risk measure as the pointwise supremum over this collection of frictionless risk measures.  Given a price process $S = (S_t)_{t = 0}^T \in \SS$ (defined in~\eqref{eq:price_process} below), we denote these frictionless multivariate risk measures as $\pi^S$.  The relationship between these frictionless multivariate risk measures and the typical (univariate) risk measures of, e.g.,~\cite{BN04,DS05,FP06,AP10}, is presented in Section~\ref{sec:pi-phi}.  These new risk measures allow us to define a new notion of time consistency, called $\pi$-time consistency and provided in Definition~\ref{defn:pi-tc} below, which overcomes the obstacles related to market frictions in both $\rho$-time consistency and multiportfolio time consistency.  Briefly, if every no-arbitrage price $S$ consistent with the market model $\K$ induces a risk measure $\pi^S$ that is time consistent in the usual way, then we say the original risk measure $\rho$ is $\pi$-time consistent.  Notably, the superhedging risk measure \emph{is} $\pi$-time consistent.  Results on $\pi$-time consistency and its relation to multiportfolio time consistency of the underlying set-valued risk measure (appearing in~\eqref{eq_scalar} and with possibly different choices of eligible assets) are provided in Sections~\ref{sec:pi-tc} and~\ref{sec:mptc}.  Though the proposed time consistency notion is weaker than $\rho$-time consistency, it is stronger than the weak notion of $\rho$-time consistency (i.e., where $\rho_s(X) \leq \rho_s(\zero)$ implies $\rho_t(X) \leq \rho_t(\zero)$ for every portfolio $X \in \LdiF{}$) as studied for univariate risk measures in, e.g.,~\cite{weber2006distribution}. Section~\ref{sec:mptc} also contains a result on a backward recursion for $\pi$-time consistent scalar multivariate risk measures $\rho_t$  with respect to the underlying risk measure $\pi^S$.

\subsection{Preliminaries}\label{sec_tc-prelim}

For notational simplicity we will introduce the set of eligible pricing processes:
\begin{equation}
\label{eq:price_process}
\SS := \lrcurly{S \in \prod_{t = 0}^T \LdiK{t}{\plus{K_t}} \; \left| \; \begin{array}{l} \forall t = 0,...,T, \; \forall i \in \{2,3,...,d\}: \\ S_{t,1} = 1, \; \|S_{t,i}\|_\infty \leq \max\{\|e_i\|_{\K_T,0} , 1\} \end{array} \right.}.
\end{equation}

In the following definition we introduce a decomposition of the risk measures $\rho_t$.  This section will introduce some simple properties of these functions, which will then be used to define a new property for time consistency.
\begin{definition}
Let $S \in \SS$.  Define $\pi_t^S: \LdiF{} \to \LzF{t} \cup \{-\infty\}$ by
\begin{equation}
\label{eq:pi}
\pi_t^S(X) := \esssup_{\Q \in \QQ^e(S)} \lrparen{-\beta_t(\Q,S) - \EQt{\trans{S_T}X}{t}}
\end{equation}
for any $X \in \LdiF{}$ where $\beta_t$ is defined as in Theorem~\ref{thm_dual_e1} dependent on the acceptance set $A_t$ and
\[\QQ^e(S) := \lrcurly{\Q \in \mathcal{M} \; | \; (\Q,S) \in \QQ_0^e}.\]
\end{definition}
Note that $\pi_t^S(X) \in \LzF{t} \cup \{-\infty\}$ in general even though $\rho_t(X) \in \LiF{t}$.  Though $\pi_t^S(X)$ is trivially bounded from above by $\rho_t(X)$, it is not integrable in general. Proposition~\ref{prop:pi_inf} provides a simple condition so that $\pi_t^S(X) \in \LiF{t}$ for any $X \in \LdiF{}$.

\begin{definition}
\label{defn:NA}
A sequence of prices $S \in \SS$ satisfies the \textbf{\emph{no-arbitrage condition}} if
\begin{equation}
\label{NA}
\tag{NA} \lrsquare{\sum_{s = 0}^{T} \Gamma_s(S_s)} \cap \lrsquare{-\Gamma_T(S_T)} \subseteq \Gamma_T(S_T)
\end{equation}
where $\Gamma_t(S_t) := \{Z \in \LdiF{t} \; | \; \trans{S_t}Z \geq 0 \sas\}$.
\end{definition}

The motivation for calling this property ``no-arbitrage'' comes from the following proposition which relates the condition to the fundamental theorem of asset pricing.

\begin{proposition}
\label{prop:NA}
Let $S \in \SS$, then $S$ satisfies the no-arbitrage condition if and only if $\QQ^e(S) \neq \emptyset$.
\end{proposition}
\begin{proof}
This follows directly from the fundamental theorem of asset pricing by noting that our definition of no-arbitrage is equivalent to the usual one in frictionless markets.  This can be seen by:
\begin{enumerate}
\item A predictable process $\seq{\eta}$ is self-financing if and only if $\trans{S_t}(\eta_{t+1} - \eta_t) \leq 0$ almost surely for every time $t = 0,1,...,T-1$ with $\eta_0 = 0$.  Equivalently this can be written as $\eta_{t+1} - \eta_t \in -\Gamma_t(S_t)$ for every time $t = 0,1,...,T-1$ with $\eta_0 = 0$.  Through summations, $\eta_t$ is attainable through self-financing trading strategy if and only if $\eta_t \in -\sum_{s = 0}^{t-1} \Gamma_s(S_s)$.
\item The wealth of a portfolio $\eta_t$ at time $t$ is given by $V_t := \trans{S_t}\eta_t$.  Therefore, in our notation, wealth satisfies the conditions $V_t \geq 0$ and $\P(V_t > 0) > 0$ if and only if $\eta_t \in \Gamma_t(S_t)$ and $\eta_t \not\in -\Gamma_t(S_t)$ respectively.
\end{enumerate}
Therefore it is immediately clear that \eqref{NA} is equivalent to stating that $V_0 = 0$, $V_T \geq 0$ almost surely implies $V_T = 0$ almost surely.  Finally, we can replace the self-financing condition $-\sum_{s = 0}^{T-1} \Gamma_s(S_s)$ with the summation up to time $T$ as these are equivalent conditions:
\begin{enumerate}
\item Assume $\lrsquare{\sum_{s = 0}^{T} \Gamma_s(S_s)} \cap \lrsquare{-\Gamma_T(S_T)} \subseteq \Gamma_T(S_T)$ then immediately by $0 \in \Gamma_T(S_T)$ we can conclude $\lrsquare{\sum_{s = 0}^{T-1} \Gamma_s(S_s)} \cap \lrsquare{-\Gamma_T(S_T)} \subseteq \Gamma_T(S_T)$.
\item Conversely, assume $\lrsquare{\sum_{s = 0}^{T-1} \Gamma_s(S_s)} \cap \lrsquare{-\Gamma_T(S_T)} \subseteq \Gamma_T(S_T)$ and let $X \in \lrsquare{\sum_{s = 0}^{T} \Gamma_s(S_s)} \cap \lrsquare{-\Gamma_T(S_T)}$. We can decompose $X = \sum_{s = 0}^T X_s$ where $\trans{S_s}X_s \geq 0$ almost surely and $\trans{S_T}X \leq 0$ almost surely.  Then we can easily see that
\begin{equation}\label{eq:NA-proof}
0 \geq \trans{S_T} X = \trans{S_T}\sum_{s = 0}^{T-1} X_s + \trans{S_T} X_T \geq \trans{S_T}\sum_{s = 0}^{T-1} X_s.
\end{equation}
Therefore $\sum_{s = 0}^{T-1} X_s \in \lrsquare{\sum_{s = 0}^{T-1} \Gamma_s(S_s)} \cap \lrsquare{-\Gamma_T(S_T)} \subseteq \Gamma_T(S_T)$, i.e., $\trans{S_T}\sum_{s = 0}^{T-1} X_s = 0$ almost surely.  This immediately implies that $\trans{S_T}X = 0$ almost surely as well from \eqref{eq:NA-proof}, i.e., $X \in \Gamma_T(S_T)$.
\end{enumerate}
\end{proof}

\begin{proposition}
\label{prop:pi0}
Let $S \in \SS$ then the following are equivalent:
\begin{enumerate}
\item \label{prop:pi0-1} $\pi_t^S(\zero) \in \LiF{t}$,
\item \label{prop:pi0-2} there exists some $\Q \in \QQ^e(S)$ such that $\beta_t(\Q,S) \in \LiF{t}$, and 
\item \label{prop:pi0-3} $S$ satisfies \eqref{NA} and there exists some $X \in \LdiF{t}$ such that $\P\lrparen{X \in \tilde A_t^\pi} = 0$, where $A_t^\pi := \lrcurly{X \in \LdiF{} \; | \; \pi_t^S(X) \leq 0} = \cl\lrparen{A_t + \sum_{s = t}^T \Gamma_s(S_s)}$ and $\tilde A_t^\pi$ is a $\Ft{t}$-measurable random set such that $A_t^\pi \cap \LdiF{t} = \LdpK{\infty}{t}{\tilde A_t^\pi}$.
\end{enumerate}
\end{proposition}
\begin{proof}
First we will show that (\ref{prop:pi0-1}.) is equivalent to (\ref{prop:pi0-2}.), then we will show this is equivalent to (\ref{prop:pi0-3}.).
\begin{enumerate}
\item First we will assume there exists some $\Q \in \QQ^e(S)$ such that $\beta_t(\Q,S) \in \LiF{t}$. For such a $\Q$ and for $\rho_t$ defined by the same acceptance set $A_t$,
\[\LiF{t} \ni -\beta_t(\Q,S) \leq \pi_t^S(\zero) \leq \rho_t(\zero) \in \LiF{t}.\]
Now we wish to prove the converse. We can trivially deduce a lower bound $\beta_t(\Q,S) \geq -\pi_t^S(\zero) \in \LiF{t}$ for every $\Q \in \QQ^e(S)$.  To prove the upper bound we first note that $\{\beta_t(\Q,S) \; | \; \Q \in \QQ^e(S)\}$ is an $\Ft{t}$-decomposable set. This implies there exists a sequence $(\Q_n)_{n \in \N} \subseteq \QQ^e(S)$ such that $-\beta_t(\Q_n,S) \nearrow \pi_t^S(\zero)$ almost surely.  Let us define $\tau^\epsilon := \min\{n \in \N \; | \; -\beta_t(\Q_n,S) + \epsilon > \pi_t^S(\zero)\} < \infty$ almost surely.  Then $\tilde\Q := \sum_{n = 1}^\infty 1_{\{\tau^\epsilon = n\}} \Q_n \in \QQ^e(S)$ (trivially utilizing monotone convergence) and 
\[\beta_t(\tilde\Q,S) \leq \sum_{n = 1}^\infty 1_{\{\tau^\epsilon = n\}} \beta_t(\Q_n,S) < -\pi_t^S(\zero) + \epsilon \in \LiF{t}.\]
\item First, $\tilde A_t^\pi$ exists by~\cite[Theorem 2.1.6]{M05} since $A_t^\pi$ is a nonempty, closed, and $\Ft{t}$-decomposable.

Given the existence of $\tilde A_t^\pi$, first we will assume that $S$ satisfies \eqref{NA} and there exists some $X \in \LdiF{t}$ such that $\P\lrparen{X \in \tilde A_t^\pi} = 0$.  Fix $X \in \LdiF{t}$ such that $\P\lrparen{X \in \tilde A_t^\pi} = 0$, then there exists some $\Q \in \QQ^e(S)$ (where $\QQ^e(S) \neq \emptyset$ by Proposition~\ref{prop:NA}) such that $\beta(\Q,S) = \esssup_{Z \in A_t} -\EQt{\trans{S_T}Z}{t} < -\trans{S_t}X \in \LiF{t}$ by a conditional separating hyperplane argument.
Additionally, $\beta_t(\Q,S) \geq -\rho_t(\zero) \in \LiF{t}$ for any $\Q \in \QQ^e(S)$ by construction.
Conversely, $\QQ^e(S) \neq \emptyset$ implies the prices $S$ satisfy \eqref{NA} by Proposition~\ref{prop:NA}.  And given $\beta_t(\Q,S) = \esssup_{Z \in A_t} -\EQt{\trans{S_T}Z}{t} \in \LiF{t}$, we can immediately conclude that $-\lrparen{\beta_t(\Q,S) + \epsilon}e_1 \in \LdiF{t}$ for any $\epsilon > 0$ and $\P\lrparen{-\lrparen{\beta_t(\Q,S) + \epsilon}e_1 \in \tilde A_t^\pi} = 0$.
\end{enumerate}
\end{proof}

\begin{proposition}
\label{prop:pi}
Fix $S \in \SS$ satisfying \eqref{NA}, the following properties hold for $\pi_t^S$:
\begin{enumerate}
\item \emph{Conditional cash invariance}: $\pi_t^S(X + m) = \pi_t^S(X) - \trans{S_t}m$ for any $X \in \LdiF{}$ and $m \in \LdiF{t}$;
\item \emph{Monotonicity}: $\pi_t^S(X) \leq \pi_t^S(Y)$ if $\trans{S_T}X \geq \trans{S_T}Y$ almost surely;
\item \emph{Conditional convexity}: $\pi_t^S(\lambda X + (1-\lambda)Y) \leq \lambda \pi_t^S(X) + (1-\lambda)\pi_t^S(Y)$ for every $\lambda \in \LpK{\infty}{t}{[0,1]}$;
\item \emph{Value consistent}: $\pi_t^S(X) = \pi_t^S([\trans{S_T}X]e_1)$ for any $X \in \LdiF{}$.
\end{enumerate}
If $A_t$ is additionally a conditional cone then $\pi_t^S$ is conditional positive homogeneous.
\end{proposition}
\begin{proof}
Trivially from the definition of $\pi_t^S$ and that $S_t = \EQt{S_T}{t}$ for any $\Q \in \QQ^e(S)$ and $S_{T,1} = 1$ almost surely.
\end{proof}

The following proposition provides a necessary and sufficient condition for $\pi_t^S(X) \in \LiF{t}$ for every $X \in \LdiF{}$.  Proposition~\ref{prop:pi0} relates this condition to the no-arbitrage condition~\eqref{NA} for the frictionless prices $S$.
\begin{proposition}
\label{prop:pi_inf}
Let $S \in \SS$ and $X \in \LdiF{}$ then, for any time $t$, $\pi_t^S(X) \in \LiF{t}$ if and only if  $\pi_t^S(\zero) \in \LiF{t}$ for any time $t$.
\end{proposition}
\begin{proof}
Let $\|X\| = \transp{\|X_1\|_\infty,\dots,\|X_d\|_\infty}$ and assume $\pi_t^S(\zero) \in \LiF{t}$. 
\begin{align*}
\pi_t^S(X) &\leq \pi_t^S(-\|X\|) = \pi_t^S(\zero) + \trans{S_t}\|X\| \in \LiF{t}\\
\pi_t^S(X) &\geq \pi_t^S(\|X\|) = \pi_t^S(\zero) - \trans{S_t}\|X\| \in \LiF{t}.
\end{align*}
The converse direction is trivial from the same inequalities.
\end{proof}

For the remainder of this work, we will use the notation 
\[\SS_t := \{S \in \SS \; | \; \pi_t^S(\zero) \in \LiF{t}\},\] 
i.e., $S \in \SS_t$ if and only if $S$ satisfies the no-arbitrage condition \eqref{NA} and there exists some $X \in \LdiF{}$ such that $\P\lrparen{X \in \tilde A_t^\pi} = 0$.

The following proposition shows the decomposition of the dual representation \eqref{eq_dual_e1} of $\rho_t$ into two components based on the two dual variables $(\Q,S)$.  $\pi_t^S(X)$, as defined in \eqref{eq:pi}, runs over the dual variables $\Q$ only.  But, it is shown in Proposition~\ref{prop:pi_to_rho} below, it is enough to take the essential supremum with respect to $S$ over the set $\SS_t$ only. This means it is enough to consider in the dual representation of $\rho_t$ just those $\pi_t^S(X)$ that map into $\LiF{t}$, i.e., to consider those frictionless prices that satisfy the no-arbitrage condition~\eqref{NA}.  This will be of particular importance when considering the new time consistency concept presented in Sections~\ref{sec:pi-tc} and~\ref{sec:mptc}.

\begin{proposition}
\label{prop:pi_to_rho} 
For any time $t$ and any $X \in \LdiF{}$ 
\[\rho_t(X) = \esssup_{S \in \SS_t} \pi_t^S(X).\]
\end{proposition}
\begin{proof}
Let $S \in \SS$ then $\Q \in \QQ^e(S)$ implies $(\Q^t,(S_s)_{s = t}^T) \in \QQ_t^e$ where $\frac{d\Q^t}{d\P} := \bar\xi_{t,T}(\Q)$.  By the definition of $\Q^t$ and the $\P$-almost sure definition of the conditional expectation, we recover that $-\beta_t(\Q,S) - \EQt{\trans{S_T}X}{t} = -\beta_t(\Q^t,S) - \EPt{\Q^t}{\trans{S_T}X}{t}$ for any $X \in \LdiF{}$.  
Conversely, let $(\Q,(S_s)_{s = t}^T) \in \QQ_t^e$ and define $\hat S \in \SS$ by $\hat S_s := \EQt{S_T}{s} = \begin{cases} S_s &\text{if } s \geq t \\ \Et{S_t}{s} &\text{if } s < t\end{cases}$.  Immediately it is clear that $\Q \in \QQ^e(\hat S)$.  By definition of $\hat S$, we recover that $-\beta_t(\Q,S) - \EQt{\trans{S_T}X}{t} = -\beta_t(\Q,\hat S) - \EQt{\trans{\hat S_T}X}{t}$ for any $X \in \LdiF{}$.

Using this result, we determine the equality of the first and second lines below:
\begin{align*}
\rho_t(X) &= \esssup_{(\Q,S) \in \QQ_t^e} \lrparen{-\beta_t(\Q,S) - \EQt{\trans{S_T}X}{t}}\\
&= \esssup_{S \in \SS} \esssup_{\Q \in \QQ^e(S)} \lrparen{-\beta_t(\Q,S) - \EQt{\trans{S_T}X}{t}}\\ 
&= \esssup_{S \in \SS} \pi_t^S(X).
\end{align*}

From this and Proposition~\ref{prop:pi0} we immediately find 
\[\rho_t(X) \geq \esssup_{\substack{S \in \SS\\ \pi_t^S(\zero) \in \LiF{t}}} \pi_t^S(X) \geq \esssup_{\substack{(\Q,S) \in \QQ_t^e\\ \beta_t(\Q,S) \in \LiF{t}}} \lrparen{-\beta_t(\Q,S) - \EQt{\trans{S_T}X}{t}}.\]
We now wish to show the reverse inequality.  Since $\{-\beta_t(\Q,S) - \EQt{\trans{S_T}X}{t} \; | \; (\Q,S) \in \QQ_t^e\}$ is $\Ft{t}$-decomposable we can find a sequence $(\Q_n,S_n)_{n \in \N} \subseteq \QQ_t^e$ such that $-\beta_t(\Q_n,S_n) - \EQnt{n}{\trans{S_{n,T}}X}{t} \nearrow \rho_t(X) \in \LiF{t}$ almost surely.  Define $\tau^\epsilon := \min\{n \in \N \; | \; \beta_t(\Q_n,S_n) + \EQnt{n}{\trans{S_{n,T}}X}{t} < -\rho_t(X) + \epsilon\} < \infty$ almost surely for any $\epsilon > 0$, and let $(\tilde\Q,\tilde{S}) := \sum_{n = 1}^\infty 1_{\{\tau^\epsilon = n\}} (\Q_n,S_n) \in \QQ_t^e$ (trivially using monotone convergence).  Then
\[\LiF{t} \ni -\rho_t(\zero) \leq \beta_t(\tilde\Q,\tilde{S}) \leq -\rho_t(X) - \EtQt{\trans{\tilde{S}_T}X}{t} + \epsilon \in \LiF{t}.\]
Since such a $(\tilde\Q,\tilde{S}) \in \QQ_t^e$ can be found for any $\epsilon > 0$, a sequence of those can be used to get convergence as $\epsilon \to 0$ and the proof is complete.
\end{proof}

\subsection{Relation to univariate scalar risk measures}\label{sec:pi-phi}
In this section we wish to relate the terms $\seq{\pi^S}$ for a fixed $S \in \SS$ introduced in the prior section to the univariate dynamic risk measures as studied in, e.g., \cite{BN04,DS05,FP06,AP10}.  This relationship will then be used later to directly prove some properties for $\seq{\pi^S}$.

\begin{proposition}
\label{prop:phi}
Let $S \in \SS_t$, then $\phi_t^S: \LiF{} \to \LiF{t}$ defined by 
\begin{equation}
\label{eq_phi}
	\phi_t^S(Z) := \pi_t^S(Ze_1)
\end{equation}
for any $Z \in \LiF{}$ is a convex (coherent) univariate scalar risk measure.  Additionally, $\pi_t^S(X) = \phi_t^S(\trans{S_T}X)$ for any $X \in \LdiF{}$.
\end{proposition}
\begin{proof}
First we will show that $\phi_t^S$ is a traditional scalar risk measure:
\begin{enumerate}
\item $\phi_t^S(0) = \pi_t^S(\zero) \in \LiF{t}$.
\item Let $Z \in \LiF{}$ and $m \in \LiF{t}$, then $\phi_t^S(Z + m) = \pi_t^S([Z + m]e_1) = \pi_t^S(Ze_1) - (\trans{S_t}e_1)m = \phi_t^S(Z) - m$.
\item Let $Z_1,Z_2 \in \LiF{}$ where $Z_1 \geq Z_2$ almost surely.  Then $Z_1e_1 \geq Z_2e_1$ almost surely, and by definition $\phi_t^S(Z_1) = \pi_t^S(Z_1e_1) \leq \pi_t^S(Z_2e_1) = \phi_t^S(Z_2)$.
\item Convexity and positive homogeneity follow immediately.
\end{enumerate}
The final statement follows by $\pi_t^S(X) = \pi_t^S([\trans{S_T}X]e_1) = \phi_t^S(\trans{S_T}X)$ for any $X \in \LdiF{}$.
\end{proof}

We will now consider the primal and dual representations for $\phi_t^S$, defined in \eqref{eq_phi}.  Afterwards we will consider the representations for the stepped risk measures $\phi_{t,s}^S$ for $t < s$.  These stepped risk measures are useful for equivalent properties of time consistency, as discussed in the next section.
\begin{corollary}
\label{cor:phi}
Let $S \in \SS_t$, then the primal and dual representations for $\phi_t^S$, defined in \eqref{eq_phi}, are given by:
\begin{enumerate}
\item $\phi_t^S(Z) = \essinf\lrcurly{m \in \LiF{t} \; | \; Z + m \in \trans{S_T}A_t^\pi}$ for any $Z \in \LiF{}$ where $A_t^\pi = \{X \in \LdiF{} \; | \; \pi_t^S(X) \leq 0\}$;
\item $\phi_t^S(Z) = \esssup_{\Q \in \QQ^e(S)} \lrparen{-\beta_t(\Q,S) - \EQt{Z}{t}}$ for any $Z \in \LiF{}$.
\end{enumerate}
\end{corollary}
\begin{proof}
First, since $\phi_t$ is a traditional scalar risk measure, we immediately know a primal and dual representation exist: for any $Z \in \LiF{}$
\begin{align*}
\phi_t^S(Z) &= \essinf\lrcurly{m \in \LiF{t} \; | \; Z + m \in A_t^\phi}\\
&= \esssup_{\Q \in \mathcal{M}} \lrparen{-\beta_t^S(\Q) - \EQt{Z}{t}}.
\end{align*}
where $A_t^\phi = \lrcurly{Z \in \LiF{} \; | \; \phi_t^S(Z) \leq 0}$ and $\beta_t^S(\Q) = \esssup_{Z \in A_t^\phi} -\EQt{Z}{t}$ for every $\Q \in \mathcal{M}$.
Therefore showing $A_t^\phi = \trans{S_T}A_t^\pi$ is sufficient to prove this result.  We note that $A_t^\pi = \ascl\lrsquare{A_t + \sum_{s = t}^T \Gamma_s(S_s)}$ in order to find
\begin{align}
\nonumber A_t^\phi &= \lrcurly{Z \in \LiF{} \; | \; \phi_t^S(Z) \leq 0}\\
\nonumber &= \lrcurly{Z \in \LiF{} \; | \; \pi_t^S(Ze_1) \leq 0}\\
\nonumber &= \lrcurly{Z \in \LiF{} \; | \; Ze_1 \in \ascl\lrsquare{A_t + \sum_{s = t}^T \Gamma_s(S_s)}}\\
\label{eq:cor:phi_proof1} &= \lrcurly{\trans{S_T}X \; | \; X \in \ascl\lrsquare{A_t + \sum_{s = t}^T \Gamma_s(S_s)}\cap M_T}\\
\label{eq:cor:phi_proof2} &= \lrcurly{\trans{S_T}X \; | \; X \in \ascl\lrsquare{A_t + \sum_{s = t}^T \Gamma_s(S_s)}}\\
\nonumber &= \trans{S_T}A_t^\pi. 
\end{align}
Equation~\eqref{eq:cor:phi_proof1} follows from $S_{T,1} = 1$ almost surely for any $S \in \SS_t$.  Equation~\eqref{eq:cor:phi_proof2} follows since $X \in A_t^\pi$ if and only if $[\trans{S_T}X]e_1 \in A_t^\pi$.

We conclude this proof with a consideration of the structure of the penalty function $\beta_t^S$.  Let $\Q \in \mathcal{M}$, then
\begin{align*}
\beta_t^S(\Q) &= \esssup_{Z \in A_t^\phi} -\EQt{Z}{t}\\
&= \esssup_{X \in A_t^\pi} -\EQt{\trans{S_T}X}{t}\\
&= \esssup_{X \in A_t} -\EQt{\trans{S_T}X}{t} + \sum_{s = t}^T \esssup_{k_s \in \Gamma_s(S_s)} -\EQt{\trans{S_T}k_s}{t}\\
&= \begin{cases} \beta_t(\Q,S) &\text{on } \{\Q \in \QQ^e(S)\}\\ \infty &\text{else}\end{cases}
\end{align*}
Since the dual representation is taken on the set with $\E{\beta_t^S(\Q)} < \infty$ this implies
\[\phi_t^S(Z) = \esssup_{\Q \in \QQ^e(S)} \lrparen{-\beta_t(\Q,S) - \EQt{Z}{t}}\] 
for any $Z \in \LiF{}$.
\end{proof}

\begin{corollary}
\label{cor:phi-stepped}
Let $S \in \SS_t$ and $t < s$. Define $\phi_{t,s}^S: \LiF{s} \to \LiF{t}$ as the restriction of the domain of $\phi_t^S$ to $\LiF{s}$.  Then primal and dual representations are defined via the acceptance set and penalty function for primal and dual representations given in:
\begin{enumerate}
\item $A_{t,s}^\phi = \trans{S_s}A_{t,s}^\pi$ where $A_{t,s}^\pi := A_t^\pi \cap \LdiF{s}$;
\item $\beta_{t,s}^S(\Q) = \esssup_{X \in A_{t,s}^\pi} -\EQt{\trans{S_s}X}{t}$ for any $\Q \in \QQ^e(S)$.
\end{enumerate}
\end{corollary}
\begin{proof}
Since $\phi_{t,s}^S$ is a univariate stepped risk measure, if we prove the result for the acceptance set then the penalty function follows immediately.  First we will show that $A_{t,s}^\phi := A_t^\phi \cap \LiF{s} \subseteq \trans{S_s}A_{t,s}^\pi$ where $A_{t,s}^\pi := A_t^\pi \cap \LdiF{s}$.  Take $Z \in A_{t,s}^\phi$, then $\pi_t^S(Ze_1) = \phi_t^S(Z) \leq 0$ thus $Ze_1 \in A_{t,s}^\pi$.  By $S_{s,1} = 1$ almost surely, $Z = \trans{S_s}(Ze_1) \in \trans{S_s}A_{t,s}^\pi$.  Now we will show that $A_{t,s}^\phi \supseteq \trans{S_s}A_{t,s}^\pi$.  Let $X \in A_{t,s}^\pi$, then
\[\phi_{t,s}^S(\trans{S_s}X) = \phi_t^S(\trans{S_s}X) = \pi_t^S([\trans{S_s}X]e_1) = \pi_t^S(X) \leq 0\]
where the last equality follows directly from the definition of $\pi_t^S$ and $X \in \LdiF{s}$.
\end{proof}

\begin{remark}
We would like to note that, unlike in the full case $\phi_t^S$, the stepped risk measure need \emph{not} satisfy the dual representation with respect to the penalty function from the scalarized conditional convex risk measure $\rho_t$.  That is, it is possible that $\phi_{t,s}^S(Z) \neq \esssup_{\Q \in \QQ^e(S)} \lrparen{-\beta_{t,s}(\Q,S) - \EQt{Z}{t}}$ for some $Z \in \LiF{s}$.
\end{remark}

\subsection{$\pi$-time consistency}\label{sec:pi-tc}
We will now introduce a new notion of time consistency for $\seq{\pi^S}$.  Using the above mentioned relations to univariate risk measures, we can deduce many properties for this time consistency property directly.

\begin{definition}\label{defn:pi-tc}
Fix $S \in \SS$.  We say $\seq{\pi^S}$ is \textbf{\emph{time consistent}} if for all times $t < s$ and $X,Y \in \LdiF{}$ it holds
\[\pi_s^S(X) \leq \pi_s^S(Y) \Rightarrow \pi_t^S(X) \leq \pi_t^S(Y).\]
The dynamic risk measure $\seq{\rho}$ is called \textbf{\emph{$\pi$-time consistent}} if $\seq{\pi^S}$ is time consistent for every price sequence $S \in \SS^* := \bigcap_{t \geq 0} \SS_t$.
\end{definition}

\begin{proposition}
\label{prop:pi-phi_tc}
Fix $S \in \SS^*$, then $\seq{\pi^S}$ is time consistent if and only if $\seq{\phi^S}$ is time consistent (defined in the usual way for univariate scalar risk measures).
\end{proposition}
\begin{proof}
Let $\seq{\pi^S}$ be time consistent and let $Z_1,Z_2 \in \LiF{}$ such that $\phi_s^S(Z_1) \geq \phi_s^S(Z_2)$ for times $t < s$:
\begin{align*}
\phi_s^S(Z_1) \geq \phi_s^S(Z_2) &\Rightarrow \pi_s^S(Z_1e_1) \geq \pi_s^S(Z_2e_1)\\
&\Rightarrow \pi_t^S(Z_1e_1) \geq \pi_t^S(Z_2e_1) \Rightarrow \phi_t^S(Z_1) \geq \phi_t^S(Z_2).
\end{align*}

Conversely, let $\seq{\phi^S}$ be time consistent and let $X_1,X_2 \in \LdiF{}$ such that $\pi_s^S(X_1) \geq \pi_s^S(X_2)$:
\begin{align*}
\pi_s^S(X_1) \geq \pi_s^S(X_2) &\Rightarrow \phi_s^S(\trans{S_T}X_1) \geq \phi_s^S(\trans{S_T}X_2)\\
&\Rightarrow \phi_t^S(\trans{S_T}X_1) \geq \phi_t^S(\trans{S_T}X_2) \Rightarrow \pi_t^S(X_1) \geq \pi_t^S(X_2).
\end{align*}
\end{proof}

We will now use the equivalence between the time consistency of $\seq{\pi^S}$ and $\seq{\phi^S}$ to determine equivalent properties for time consistency of $\seq{\pi^S}$. 
\begin{theorem}\label{thm:pi-tc}
Fix $S \in \SS^*$, then the following are equivalent for a normalized and convex $\seq{\pi^S}$ (i.e., $\pi_t^S(\zero) = 0$ for all times $t \geq 0$).
\begin{enumerate}
\item $\seq{\pi^S}$ is time consistent;
\item $\pi_t^S(X) = \pi_t^S(-\pi_s^S(X)e_1)$ for all $t < s$ and $X \in \LdiF{}$; 
\item $A_t^\pi = A_{t,s}^\pi + A_s^\pi$ where $A_{t,s}^\pi = A_t^\pi \cap \LdiF{s}$ for all $t < s$;
\item $\beta_t(\Q,S) = \beta_{t,s}^S(\Q) + \EQt{\beta_s(\Q,S)}{t}$ for all $t < s$ and all probability measures $\Q \in \QQ^*(S)$, where we define
\begin{equation}
\label{QQ}
\QQ^*(S) := \lrcurly{\Q \in \QQ^e(S) \; | \; \beta_0(\Q,S) < \infty};
\end{equation}
\item For all $\Q \in \QQ^*(S)$ and all $X \in \LdiF{}$, the process
\[V_t^\Q(X) := \pi_t^S(X) + \beta_t(\Q,S), \quad t \geq 0\]
is a $\Q$-supermartingale.
\end{enumerate}
In each case the dynamic risk measure admits a robust representation in terms of the set $\QQ^*(S)$ defined in \eqref{QQ}, i.e.,
\[\pi_t^S(X) = \esssup_{\Q \in \QQ^*(S)} \lrparen{-\beta_t(\Q,S) - \EQt{\trans{S_T}X}{t}}\]
for all $X \in \LdiF{}$ and all times $t \geq 0$.
\end{theorem}
\begin{proof}
These results all follow immediately from the equivalent properties of $\seq{\phi^S}$-time consistency and its equivalence to $\seq{\pi^S}$-time consistency (see Proposition~\ref{prop:pi-phi_tc}).  To show the equivalence of the recursive formulation and the summation of acceptance sets we consider Lemma~\ref{lemma:pi-tc-acceptance} below as the equivalence between acceptance set versions of time consistency for $\seq{\phi^S}$ and $\seq{\pi^S}$ is less direct.
\end{proof}

\begin{remark}\label{rem:pi-tc}
Consider the setting of Theorem~\ref{thm:pi-tc} but for a non-normalized $\seq{\pi^S}$.  Then $\seq{\pi^S}$ is time consistent if and only if $\pi_t^S(X) = \pi_t^S([\pi_s^S(\zero)-\pi_s^S(X)]e_1)$ for all $t < s$ and $X \in \LdiF{}$.
\end{remark}

Now consider coherent $\seq{\pi^S}$ only.  Note that, immediately, the coherence of $\seq{\pi^S}$ implies normalization.  Additionally, in such a setting, the penalty function is provided by an indicator function; as such we can describe the probability measures of interest, defined in \eqref{QQ}, by
\[\QQ^*(S) = \lrcurly{\Q \in \QQ^e(S) \; | \; \beta_0(\Q,S) = 0}.\]
\begin{corollary}\label{cor:pi-tc}
Fix $S \in \SS^*$, then the following are equivalent for a coherent $\seq{\pi^S}$.
\begin{enumerate}
\item $\seq{\pi^S}$ is time consistent;
\item $\pi_t^S(X) = \pi_t^S(-\pi_s^S(X)e_1)$ for all $t < s$ and $X \in \LdiF{}$;
\item $A_t^\pi = A_{t,s}^\pi + A_s^\pi$ where $A_{t,s}^\pi = A_t^\pi \cap \LdiF{s}$ for all $t < s$;
\item The set $\QQ^*(S)$ is stable;
\item For all $\Q \in \lrcurly{\Q \in \QQ^e(S) \; | \; \beta_0(\Q,S) = 0}$ and all $X \in \LdiF{}$, the process $\lrparen{\pi_t^S(X)}_{t = 0}^T$ is a $\Q$-supermartingale.
\end{enumerate}
In each case the dynamic risk measure admits a robust representation in terms of the set $\QQ^*(S)$, i.e.,
\[\pi_t^S(X) = \esssup_{\Q \in \QQ^*(S)} -\EQt{\trans{S_T}X}{t}\]
for all $X \in \LdiF{}$ and all times $t \geq 0$.
\end{corollary}

\begin{lemma}\label{lemma:pi-tc-acceptance}
Fix $S \in \SS^*$ and let $0 \leq t \leq s \leq T$.
\begin{enumerate}
\item Let $D \subseteq \LdiF{s}$ such that $D + \Gamma_s(S_s) \subseteq D$.  Then $X \in D + A_s^\pi$ if and only if $-\pi_s^S(X) \in \trans{S_s}D$.
\item $A_{t,s}^\phi + A_s^\phi = \trans{S_T}\lrsquare{A_{t,s}^\pi + A_s^\pi}$
\end{enumerate}
\end{lemma}
\begin{proof}
\begin{enumerate}
\item First let $X \in D + A_s^\pi$.  Immediately we can decompose $X = X_D + X_s$ such that $X_D \in D$ and $X_s \in A_s^\pi$.  Then $\pi_s^S(X) = \pi_s^S(X_D + X_s) = \pi_s^S(X_s) - \trans{S_s}X_D \leq -\trans{S_s}X_D$.  That is, $-\pi_s^S(X) \geq \trans{S_s}X_D$ and by property of $D + \Gamma_s(S_s) \subseteq D$ this implies $-\pi_s^S(X)e_1 \in D$ or $-\pi_s^S(X) \in \trans{S_s}D$ by $S_{s,1} = 1$.  Conversely let $-\pi_s^S(X) \in \trans{S_s}D$ for some $X \in \LdiF{}$.  We can then define $X = -\pi_s^S(X)e_1 + (X + \pi_s^S(X)e_1)$.  We conclude that $-\pi_s^S(X)e_1 \in D$ and by definition of the acceptance set we see that $X + \pi_s^S(X)e_1 \in A_s^\pi$, thus the proof is complete.
\item First consider $Z \in A_{t,s}^\phi + A_s^\phi$.  By \cite[Lemma 4.6]{FP06}, this implies $-\pi_s^S(Ze_1) = -\phi_s^S(Z) \in A_{t,s}^\phi = \trans{S_s}A_{t,s}^\pi$. Therefore, $Ze_1 \in A_{t,s}^\pi + A_s^\pi$ and thus $Z \in \trans{S_T}\lrsquare{A_{t,s}^\pi + A_s^\pi}$ by $S_{T,1} = 1$. Conversely, $X \in A_{t,s}^\pi + A_s^\pi$ if and only if $-\phi_s^S(\trans{S_T}X) = -\pi_s^S(X) \in \trans{S_s}A_{t,s}^\pi = A_{t,s}^\phi$.  By \cite[Lemma 4.6]{FP06} this is equivalent to $\trans{S_T}X \in A_{t,s}^\phi + A_s^\phi$, thus the proof is complete.
\end{enumerate}
\end{proof}

We conclude this section by considering the backwards composition to construct a $\pi$-time consistent risk measure.
\begin{proposition}
\label{prop:pi-phi_composition}
Let $S \in \SS^*$ and consider a discrete time setting.  Define the backwards composition of the portfolio returns $Z \in \LiF{}$ by
\begin{align*}
\tilde\phi_T^S(Z) &= -Z, \qquad \tilde\phi_t^S(Z) = \phi_t^S(-\tilde\phi_{t+1}^S(Z)) \; \forall t < T.
\end{align*}
Equivalently we can define the backwards composition of the portfolio $X \in \LdiF{}$ by
\begin{align*}
\tilde\pi_T^S(X) &= -\trans{S_T}X, \qquad \tilde\pi_t^S(X) = \pi_t^S(-\tilde\pi_{t+1}^S(X)e_1) \; \forall t < T.
\end{align*}
Then $\seq{\tilde\phi^S}$, equivalently $\seq{\tilde\pi^S}$, is time consistent.  
Additionally, the scalarization 
\[\tilde\rho_t(X) := \esssup_{S \in \SS^*} \tilde\pi_t^S(X)\quad \forall X \in \LdiF{}\]
is $\pi$-time consistent.
\end{proposition}
\begin{proof}
First, the time consistency of $\seq{\tilde\phi^S}$ and $\seq{\tilde\pi^S}$ follow trivially by an induction argument.
Second, $\seq{\tilde\rho}$ is $\pi$-time consistent immediately by its construction via $\{\seq{\tilde\pi^S} \; | \; S \in \SS^*\}$.
\end{proof}

\subsection{Relation to multiportfolio time consistency}\label{sec:mptc}
In this section we will relate the concept of multiportfolio time consistency defined in \cite{FR12,FR12b} for set-valued risk measures to time consistency properties and representations for scalar risk measures $\seq{\rho}$ with a single eligible asset.  

\begin{lemma}
\label{lemma:mptc-pi_tc}
Let $\seq{A}$ be a dynamic acceptance set.  Let the set-valued risk measure $\seq{R}$ (with any choice of eligible space $\tilde{M} \subseteq \R^d$), defined by
\[R_t(X) := \lrcurly{m \in \LdiK{t}{\tilde{M}} \; | \; X + m \in A_t} \quad \forall X \in \LdiF{},\]
be multiportfolio time consistent, then $\seq{\rho}$ defined as in \eqref{eq_scalar} is $\pi$-time consistent.
\end{lemma}

We want to emphasize the following point about Lemma~\ref{lemma:mptc-pi_tc}. Note that  $\seq{\rho}$ is defined with respect to the single eligible asset ($M=\R \times \{0\}^{d-1}$) by its definition via \eqref{eq_scalar}, or equivalently~\eqref{A-rep} as it is the essential infimum of the cash asset making the position acceptable. However, we now consider $\seq{R}$ with any choice of eligible space $\tilde{M} \subseteq \R^d$ (which actually does not have an impact on the definition of $\rho_t$ via \eqref{eq_scalar} or~\eqref{A-rep}, nor does it effect $A_t$).  The remarkable power of Lemma~\ref{lemma:mptc-pi_tc} is the fact that if $\seq{R}$ is multiportfolio time consistent with {\it some} choice of eligible space $\tilde{M} \subseteq \R^d$, then $\seq{\rho}$ defined as in \eqref{eq_scalar}, and thus with the single eligible asset, is $\pi$-time consistent. This result will prove to be very useful in the example section. 

\begin{proof}[Proof of Lemma~\ref{lemma:mptc-pi_tc}]
We will prove the result with full eligible space $\tilde{M} = \R^d$, the general case follows by the same logic.  Let $S \in \SS^*$, $t < s$ and $X,Y \in \LdiF{}$ 
then
\begin{align*}
&\pi_s^S(X) \leq \pi_s^S(Y)\\
&\Rightarrow \lrcurly{\trans{S_s}u \; | \; u \in \LdiF{s}, \; X + u \in A_s^\pi} \supseteq \lrcurly{\trans{S_s}u \; | \; u \in \LdiF{s}, \; Y + u \in A_s^\pi}\\
&\Rightarrow \lrcurly{u \in \LdiF{s} \; | \; X + u \in \ascl\lrsquare{A_s + \sum_{r = s}^T \Gamma_r(S_r)}}\\
&\qquad\qquad \supseteq \lrcurly{u \in \LdiF{s} \; | \; Y + u \in \ascl\lrsquare{A_s + \sum_{r = s}^T \Gamma_r(S_r)}}\\
&\Rightarrow \lrcurly{u \in \LdiF{s} \; | \; X + u \in A_s + \sum_{r = s}^T \Gamma_r(S_r)} \supseteq \lrcurly{u \in \LdiF{s} \; | \; Y + u \in A_s + \sum_{r = s}^T \Gamma_r(S_r)}\\
&\Rightarrow \lrcurly{u \in \LdiF{s} \; | \; X + u \in A_s + \sum_{r = t}^T \Gamma_r(S_r)} \supseteq \lrcurly{u \in \LdiF{s} \; | \; Y + u \in A_s + \sum_{r = t}^T \Gamma_r(S_r)}\\
&\Rightarrow \bigcup_{\hat X \in X - \sum_{r = t}^T \Gamma_r(S_r)} R_s(\hat X) \supseteq \bigcup_{\hat Y \in Y - \sum_{r = t}^T \Gamma_r(S_r)} R_s(\hat Y)\\
&\Rightarrow \bigcup_{\hat X \in X - \sum_{r = t}^T \Gamma_r(S_r)} R_t(\hat X) \supseteq \bigcup_{\hat Y \in Y - \sum_{r = t}^T \Gamma_r(S_r)} R_t(\hat Y)\\
&\Rightarrow \lrcurly{u \in \LdiF{t} \; | \; X + u \in A_t + \sum_{r = t}^T \Gamma_r(S_r)} \supseteq \lrcurly{u \in \LdiF{t} \; | \; Y + u \in A_t + \sum_{r = t}^T \Gamma_r(S_r)}\\
&\Rightarrow \pi_t^S(X) \leq \pi_t^S(Y).
\end{align*}
\end{proof}

In fact, we will now use the above result relating multiportfolio time consistency to $\pi$-time consistency in order to present a recursive definition of $\seq{\rho}$ with respect to the decomposed form $\seq{\pi^S}$.
\begin{corollary}\label{cor:rho-tc}
Consider the setting of Lemma~\ref{lemma:mptc-pi_tc} for normalized acceptance sets $\seq{A}$ (i.e., $A_t = \lrcurly{X \in \LdiF{} \; | \; \zero \in R_t(X)}$ for set-valued risk measure such that $R_t(X) = R_t(X) + R_t(\zero)$ for every $X \in \LdiF{}$ for all times $t$).  Then
$\seq{\rho}$ satisfies the recursive relation
\begin{equation}\label{eq:recursive}
\rho_t(X) = \esssup_{S \in \SS_t} \pi_t^S([\pi_s^S(\zero)-\pi_s^S(X)]e_1)
\end{equation}
for all portfolios $X \in \LdiF{}$ and times $t < s$.
\end{corollary}
\begin{proof}
This is an immediate result of Proposition~\ref{prop:pi_to_rho} and the recursive form from Theorem~\ref{thm:pi-tc} for non-normalized $\pi_s^S$ (see Remark~\ref{rem:pi-tc}).  Proposition~\ref{prop:St} is used to guarantee that $\pi_s^S(\zero)-\pi_s^S(X) \in \LiF{s}$ for every $S \in \SS_t$.
\end{proof}

Briefly, we will provide some interpretations and insights for the recursive relation \eqref{eq:recursive}.  The recursive relation for $\rho_t$ indicates that the risk measure can be related to the dynamic programming principle for $(\pi^S_t)_{t = 0}^T$ over all frictionless prices $S$ satisfying the no-arbitrage condition~\eqref{NA}.
To clarify what this does not mean, assume for simplicity $\rho_t(X) = \pi_t^{S^t_X}(X)$ for some $S^t_X \in \SS_t$ for every $X \in \LdiF{}$ and every time $t$, i.e., the essential supremum is attained at $S^t_X$.  Though it is tempting to assume this is sufficient to recover the usual recursion $\rho_t(X) = \rho_t(\rho_s(\zero)-\rho_s(X))$ equivalent to~\eqref{eq:rho-tc}, the fact that $S^t_X$ depends on both the time $t$ and portfolio $X$ provides the recursion
\[\rho_t(X) = \pi_t^{S^t_X}([\pi_s^{S^t_X}(\zero)-\pi_s^{S^t_X}(X)]e_1) \neq \pi_t^{S^t_X}([\pi_s^{S^s_{\zero}}(\zero)-\pi_s^{S^s_X}(X)]e_1) = \rho_t(\rho_s(\zero)-\rho_s(X)).\]
More generally, the essential suprema over $S \in \SS_t$ need not limit in the same direction for different times $t$ and portfolios $X$. Since one is unable to merge the essential suprema at time $t$ and $s$, the usual recursion $\rho_t(X) = \rho_t(\rho_s(\zero)-\rho_s(X))$ does \emph{not} hold.
In fact, the recursive relation~\eqref{eq:recursive} can be viewed as $(\rho_t)_{t = 0}^T$ satisfying the dynamic programming principle for $(\pi^S_t)_{t = 0}^T$ for a \emph{fixed} scalarization $S$. Recall now, that $(\rho_t)_{t = 0}^T$ itself is a \emph{fixed} scalarization of $(R_t)_{t = 0}^T$ (in the direction of the cash asset). In \cite{FR18-scalar} we show that scalarized risk measures are in general only recursive with itself if the scalarizations change over time in a certain way (see also \cite{KMZ17,KR18}). Since here the cash asset and thus the scalarization is fixed, $(\rho_t)_{t = 0}^T$ is not recursive, but by Corollary~\ref{cor:rho-tc} satisfies the weaker recursive property~\eqref{eq:recursive} with respect to a fixed scalarization of $(\pi^S_t)_{t = 0}^T$. 

The above recursive form for $\seq{\rho}$ simplifies when $\pi_s^S(\zero) = 0$ for every $S \in \SS_t$, e.g., if $\seq{\rho}$ is a coherent risk measure. The following proposition provides further results under the same setting as Corollary~\ref{cor:rho-tc} above.  In particular, we find a relation between the sets $\SS_t$ over time, and specifically we can deduce that $\SS^* = \SS_0$.  

\begin{proposition}
\label{prop:St}
Consider the setting of Lemma~\ref{lemma:mptc-pi_tc} for closed and normalized acceptance sets $\seq{A}$.  Then
$\SS_t \subseteq \SS_s$ for any times $t < s$ and $\SS^* = \SS_0$. 
\end{proposition}
\begin{proof}
Assume $S \in \SS_t$.  Let $X \in \LdiF{}$ such that $\P(X \in \tilde A_t^\pi) = 0$ (the existence of which is guaranteed by Proposition~\ref{prop:pi0}).  Then by multiportfolio time consistency we find that:
\begin{align*}
0 &= \P(X \in \tilde A_t^\pi) = \P\lrparen{X \in \widetilde{\ascl}\lrsquare{A_t + \sum_{r = t}^T \Gamma_r(S_r)}}\\
&= \P\lrparen{X \in \widetilde{\ascl}\lrsquare{A_t \cap M_s + A_s + \sum_{r = t}^T \Gamma_r(S_r)}}\\
&\geq \P\lrparen{X \in \widetilde{\ascl}\lrsquare{A_s + \sum_{r = s}^T \Gamma_r(S_r)}} = \P(X \in \tilde A_s^\pi).
\end{align*}
Therefore the choice of $X \in \LdiF{}$ also satisfies the condition such that $S \in \SS_s$. In the above, we let $\widetilde{\ascl}$ denote the operator which generates the random set associated with the closure of the set of random variables as with $\tilde A_t^\pi$ for $A_t^\pi$ defined in Proposition~\ref{prop:pi0}.  The existence of these random sets is guaranteed by \cite[Theorem 2.1.6]{M05}. Finally, by construction in Definition~\ref{defn:pi-tc}, $\SS^* = \bigcap_{t \geq 0} \SS_t = \SS_0$ by the monotonicity of these sets of frictionless price processes.
\end{proof}

We will conclude our discussion of time consistency by relating the usual definition of time consistency of the scalar mappings $\seq{\rho}$ to multiportfolio time consistency of the set-valued risk measures $\seq{R}$ under the same single eligible asset space.
\begin{lemma}\label{lemma:rho-tc}
Consider the setting of Theorem~\ref{thm_scalar}.  $\seq{R}$ is multiportfolio time consistent if and only if $\seq{\rho}$ is $\rho$-time consistent (i.e., time consistent as defined in \eqref{eq:rho-tc}).
\end{lemma}
\begin{proof}
First, by Theorem~\ref{thm_scalar}, we know that $R_t(X) = \lrparen{\rho_t(X) + \LdiK{t}{\R_+}}e_1$ for all times $t$ and portfolios $X \in \LdiF{}$.

Let $\seq{R}$ be multiportfolio time consistent. For any $t < s$ and any $X,Y \in \LdiF{}$,
\begin{align*}
\rho_{s}(X) \leq \rho_{s}(Y) &\Rightarrow R_{s}(X) \supseteq R_{s}(Y) \Rightarrow R_t(X) \supseteq R_t(Y) \Rightarrow \rho_t(X)
\leq \rho_t(Y).
\end{align*}

Let $\seq{\rho}$ be time consistent. Additionally let $t < s$, $X \in \LdiF{}$, and $\Y \subseteq \LdiF{}$,
\begin{align*}
R_{s}(X) \subseteq \bigcup_{Y \in \Y} R_{s}(Y) &\Rightarrow \lrparen{u \geq \rho_{s}(X) \Rightarrow \exists Y \in \Y:\; u \geq \rho_{s}(Y)}\\
&\Rightarrow \exists Y \in \Y:\; \rho_{s}(X) \geq \rho_{s}(Y) \Rightarrow \exists Y \in \Y:\; \rho_t(X) \geq \rho_t(Y)\\
&\Rightarrow \exists Y \in \Y:\; R_t(X) \subseteq R_t(Y) \Rightarrow R_t(X) \subseteq \bigcup_{Y \in \Y} R_t(Y)
\end{align*}
\end{proof}

Naturally, $\rho$-time consistency is a stronger property than $\pi$-time consistency.
\begin{corollary}
If $\seq{\rho}$ is $\rho$-time consistent, then it is also $\pi$-time consistent.
\end{corollary}
\begin{proof}
This follows directly from Lemmas~\ref{lemma:mptc-pi_tc} and~\ref{lemma:rho-tc}.
\end{proof}

\section{Examples}
\label{sec_ex}

Due to Lemma~\ref{lemma:mptc-pi_tc}, we immediately know that every multiportfolio time consistent set-valued risk measure generates a multivariate scalar risk measure $\seq{\rho}$ which is $\pi$-time consistent.  Thus we refer to \cite{FR12,FR12b,FR14-alg} for some examples.  In this section we will focus on results on the superhedging risk measure with a single eligible asset, followed by an example for composed risk measures like the composed average value at risk.  For these examples we will demonstrate the concept of $\pi$-time consistency and show that the prior, strong, time consistency property ($\rho_s(X) \leq \rho_s(Y) \Rightarrow \rho_t(X) \leq \rho_t(Y)$) is too strong of a concept.

\subsection{Superhedging under proportional transaction costs}
\label{sec_shp}

Let us first consider a market with proportional transaction costs only. As introduced in Section~\ref{sec_scalar} and discussed in more details in~\cite{K99,S04,KS09}, the market model will be defined by a sequence of solvency cones $\seq{K}$.  We will assume that the market satisfies the robust no-arbitrage property and Assumption~\ref{ass_friction}.  Such a superhedging risk measure was studied at time $0$ in the discrete time setting in, e.g.,~\cite{JK95,LR11}.  By the definition of the superhedging risk measure, with trading at discrete times $t \in \{0,1,...,T\}$, we can define the dual representation via
\[\rho^{SHP}_t(X) = \esssup_{(\Q,S) \in \QQ_t^{SHP}} \EQt{-\trans{S_T}X}{t}\]
where
\[\QQ_t^{SHP} := \lrcurly{(\Q,S) \in \QQ_t^e \; | \; \bar\xi_{t,s}(\Q) S_s \in \LdpK{1}{s}{K_s^+} \; \forall s \geq t}.\]
In particular, for the construction of the decomposition $\pi_t^S$ we care about
\[\QQ_t^{SHP}(S) := \begin{cases} EMM(S) := \lrcurly{\Q \in \mathcal{M}^e \; | \; S \text{ is a } \Q\text{-mtg}} & \text{if } S \in \SS
\\ \emptyset &\text{else} \end{cases}.\]
In fact, in this case $\SS_t = \SS$ for any time $t \in [0,T]$. In particular, this implies that $\SS^* = \SS$ as well.
This implies that the dual variables for the construction of $\pi_t^S$ are provided by the set of equivalent martingale measures $EMM(S)$ for the price process $S$.  With this construction we can define the acceptance set for $\phi_t^S$ (the univariate scalar risk measure corresponding with $\pi_t^S$ as defined in \eqref{eq_phi}) as
\begin{align*}
A_t^\phi &:= \lrcurly{Z \in \LiF{} \; | \; \EQt{Z}{t} \geq 0 \; \forall \Q \in \QQ_t^{SHP}(S)}\\
&= \lrcurly{Z \in \LiF{} \; | \; \EQt{Z}{t} \geq 0 \; \forall \Q \in EMM(S)}.
\end{align*}
Immediately this implies that $A_t^\phi$ is the acceptance set of the scalar superhedging risk measure, and therefore we recover that $\seq{\phi^S}$ is time consistent (and thus so is $\pi^S$ by Proposition~\ref{prop:pi-phi_tc}) as is expected since the set-valued superhedging risk measure is multiportfolio time consistent in the case with full eligible space $\tilde{M} = \R^d$ (see, e.g., \cite{FR12}).
Thus, $\seq{\rho^{SHP}}$ is $\pi$-time consistent by its definition (Definition~\ref{defn:pi-tc}), respectively by Lemma~\ref{lemma:mptc-pi_tc}.

\begin{remark}
Corollary~\ref{cor:rho-tc} provides a recursive relation for $\seq{\rho^{SHP}}$, which we would like to note is similar to, but distinct from, that utilized in~\cite[Corollary~6.2]{LR11}. 
\end{remark}

We conclude by demonstrating that the superhedging risk measure $\seq{\rho^{SHP}}$, even though it is $\pi$-time consistent, is itself not $\rho$-time consistent.  This follows by Lemma~\ref{lemma:rho-tc}. The underlying set-valued risk measure $\seq{R}$ is a closed and conditionally coherent risk measure with acceptance sets $A_t = \sum_{s = t}^T \LdpK{\infty}{s}{K_s}$ for all times $t$ and $M = \R \times \{0\}^{d-1}$.  It is clear from the definition of the acceptance sets and the space of eligible portfolios that, in general, $A_t \neq A_t \cap M_s + A_s$.  Thus, by \cite[Theorem 3.4]{FR12}, we can immediately conclude that $\seq{R}$ is not multiportfolio time consistent with this choice of eligible assets, and by Lemma~\ref{lemma:rho-tc} $\seq{\rho^{SHP}}$ is \emph{not} $\rho$-time consistent as defined in \eqref{eq:rho-tc}.

\subsection{Superhedging under convex transaction costs}
\label{sec_shp_convex}

Continuing with the superhedging example, let us now consider a market with convex transaction costs, e.g., one that also includes price impact.  As discussed in~\cite{PP10,FR12b,FR14-alg}, we will model such a market with convex solvency regions $\seq{C}$.  We will assume that this market model satisfies the no-scalable robust no-arbitrage condition and has recession cones satisfying Assumption~\ref{ass_friction}.  

As in the case with proportional transaction costs above, we can introduce the dual representation, in the form of those provided in \cite{JK95}, for the superhedging price $\rho^{SHP}_t$ via
\[\rho^{SHP}_t(X) = \esssup_{(\Q,S) \in \QQ_t^e} \lrparen{\sum_{s = t}^T \sigma_t^s(\Et{\dQdP}{s}S_s) - \EQt{\trans{S_T}X}{t}},\]
where the penalty function is defined as the support function of the solvency regions
\[\sigma_t^s(Y) := \essinf_{Z \in \LdiK{s}{C_s}} \Et{\trans{Y}Z}{t}.\]
In particular, the penalty function is finite only if $\Et{\dQdP}{s}S_s \in \LdpK{1}{s}{\plus{\recc{C_s}}}$ for all times $s \in \{t,...,T\}$.

\begin{remark}
If we consider a conical market model $\seq{K}$, then the penalty function is determined by \[\sigma_t^s(Y) = \begin{cases}0 &\text{on } \{Y \in \plus{K_s}\}\\ -\infty &\text{on } \{Y \not\in \plus{K_s}\}\end{cases}.\] Thus we immediately recover the superhedging price presented in Section~\ref{sec_shp} as a special case.  Therefore we recover the formulation from Jouini, Kallal \cite{JK95} as seen in \eqref{JK d assets} in the static case as well.
\end{remark}

As with the case under proportional transaction costs only, it immediately follows from Lemma~\ref{lemma:mptc-pi_tc} and the multiportfolio time consistency of the set-valued superhedging risk measure with full eligible space $\tilde{M} = \R^d$ (see, e.g., \cite{FR12,FR12b}) that $\seq{\rho^{SHP}}$ is $\pi$-time consistent.  

As in the case with proportional transaction costs only, though we find that the superhedging price is $\pi$-time consistent, it is not $\rho$-time consistent (with respect to definition \eqref{eq:rho-tc}).  This follows by Lemma~\ref{lemma:rho-tc} as the superhedging risk measure is not multiportfolio time consistent when the space of eligible assets $M$ is only the cash asset.

\subsection{Composed risk measures}
\label{sec_composed}

Composing the set-valued risk measures backwards in time automatically guarantees multiportfolio time consistency.  That is, consider one-step risk measures $(R_{t,t+1})_{t = 0}^{T-1}$ and a terminal risk measure $R_T$ to define the composed risk measure by:
\begin{align*}
\tilde R_t(X) &:= R_{t,t+1}(\tilde R_{t+1}(X)) = \bigcup_{Z \in \tilde R_{t+1}(X)} R_{t,t+1}(-Z) \quad \forall t \in \{0,1,...,T-1\}\\
\tilde R_T(X) &= R_T(X).
\end{align*}
This was studied for, e.g., the set-valued average value at risk in \cite{FR12,FR12b}.  As with the superhedging example above, often we consider the full space of eligible assets $\tilde{M} = \R^d$ as it is natural when considering trading in a market model \cite{FR14-alg} (as is studied in this paper) or with systemic risk measures \cite{FRW15}.  In that setting, by Lemmas~\ref{lemma:mptc-pi_tc} and~\ref{lemma:rho-tc}, the scalarization of the composed risk measure $\tilde \rho_t(X) :=\essinf\{m \in \LiF{t} \; | \; m e_1 \in \tilde R_t(X)\}$ for $X \in \LdiF{}$ is $\pi$-time consistent but not $\rho$-time consistent.

For illustrative purposes, let us briefly consider the composed average value at risk with full space of eligible assets $\tilde{M} = \R^d$.  We refer to \cite[Section 6.1]{FR12b} for details on the set-valued composed average value at risk.
In particular, consider the stepped average value at risk with levels $\lambda^t \in \LdiK{t}{[\epsilon,1]^d}$ for the stepped risk measure from time $t$ to $t+1$ and some lower threshold $\epsilon > 0$.  The scalarization $\tilde \rho_t: \LdiF{} \to \LiF{t}$ of the composed average value at risk at time $t$, defined by \eqref{eq_scalar}, can be given by the dual representation:
\begin{align*}
\tilde \rho_t(X) &= \esssup_{S \in \SS_t} \esssup_{\Q \in \tilde\QQ^\lambda(S)} -\EQt{\trans{S_T}X}{t}\\
\tilde\QQ^\lambda(S) &= \lrcurly{\Q \in \mathcal{M}^e \; | \; S_{s,i} \geq \lambda_i^s \bar\xi_{s,s+1}(\Q) S_{s+1,i} \; \forall s = 0,1,...,T-1, i = 1,2,...,d}.
\end{align*}
Note that any probability measure $\Q \in \tilde\QQ^\lambda(S)$ for some $S \in \SS_t$ (and any fixed time $t$) will also be a dual variable for the univariate composed average value at risk with sequence of levels $(\lambda_1^t)_{t = 0}^{T-1}$ as described in \cite{CK10}.  As discussed above, this multivariate risk measure is $\pi$-time consistent but not $\rho$-time consistent.

\bibliographystyle{plain}
\bibliography{biblio}

\begin{thebibliography}{10}

\bibitem{AP10}
Beatrice Acciaio and Irina Penner.
\newblock Dynamic risk measures.
\newblock In Giulia Di~Nunno and Bernt \"{O}ksendal, editors, {\em Advanced
  Mathematical Methods for Finance}, pages 1--34. Springer, 2011.

\bibitem{AD99}
Philippe Artzner, Freddy Delbaen, Jean-Marc Eber, and David Heath.
\newblock Coherent measures of risk.
\newblock {\em Mathematical Finance}, 9(3):203--228, 1999.

\bibitem{AD07}
Philippe Artzner, Freddy Delbaen, Jean-Marc Eber, David Heath, and Hyejin Ku.
\newblock Coherent multiperiod risk adjusted values and {B}ellman's principle.
\newblock {\em Annals OR}, 152(1):5--22, 2007.

\bibitem{ADM09}
Philippe Artzner, Freddy Delbaen, and Pablo Koch-Medina.
\newblock Risk measures and efficient use of capital.
\newblock {\em ASTIN Bulletin}, 39(1):101--116, 2009.

\bibitem{BMM18}
Michel Baes, Pablo Koch-Medina, and Cosimo Munari.
\newblock Existence, uniqueness and stability of optimal portfolios of eligible
  assets.
\newblock {\em Mathematical Finance}, 30(1):128--166, 2020.

\bibitem{TL12}
Imen Ben~Tahar and Emmanuel L\'{e}pinette.
\newblock Vector-valued coherent risk measure processes.
\newblock {\em International Journal of Theoretical and Applied Finance},
  17(2):1450011, 2014.

\bibitem{BLPS92}
Bernard Bensaid, Jean-Philippe Lesne, Henri Pag\`{e}s, and Jos\'{e} Scheinkman.
\newblock Derivative asset pricing with transaction costs.
\newblock {\em Mathematical Finance}, 2(2):63--86, 1992.

\bibitem{BFFM15}
Francesca Biagini, Jean-Pierre Fouque, Marco Frittelli, and Thilo
  Meyer-Brandis.
\newblock A unified approach to systemic risk measures via acceptance sets.
\newblock {\em Mathematical Finance}, 29(1):329--367, 2019.

\bibitem{BN04}
Jocelyne Bion-Nadal.
\newblock Conditional risk measures and robust representation of convex risk
  measures.
\newblock {\em Ecole Polytechnique, CMAP, preprint no. 557}, 2004.

\bibitem{BN08}
Jocelyne Bion-Nadal.
\newblock Dynamic risk measures: Time consistency and risk measures from {BMO}
  martingales.
\newblock {\em Finance and Stochastics}, 12(2):219--244, 2008.

\bibitem{BN09}
Jocelyne Bion-Nadal.
\newblock Time consistent dynamic risk processes.
\newblock {\em Stochastic Processes and their Applications}, 119(2):633--654,
  2009.

\bibitem{bjork2017time}
Tomas Bj{\"o}rk, Mariana Khapko, and Agatha Murgoci.
\newblock On time-inconsistent stochastic control in continuous time.
\newblock {\em Finance and Stochastics}, 21(2):331--360, 2017.

\bibitem{bjork2014mean}
Tomas Bj{\"o}rk, Agatha Murgoci, and Xun~Yu Zhou.
\newblock Mean--variance portfolio optimization with state-dependent risk
  aversion.
\newblock {\em Mathematical Finance: An International Journal of Mathematics,
  Statistics and Financial Economics}, 24(1):1--24, 2014.

\bibitem{BV92}
Phelim~P. Boyle and Ton Vorst.
\newblock Option replication in discrete time with transaction costs.
\newblock {\em The Journal of Finance}, 47(1):271--293, 1992.

\bibitem{BC13}
Markus~K. Brunnermeier and Patrick Cheridito.
\newblock Measuring and allocating systemic risk.
\newblock {\em Risks}, 7(2):46, 2019.

\bibitem{BR06}
Christian Burgert and Ludger R\"{u}schendorf.
\newblock Consistent risk measures for portfolio vectors.
\newblock {\em Insurance: Mathematics and Economics}, 38(2):289--297, 2006.

\bibitem{CH17}
Yanhong Chen and Yijun Hu.
\newblock Time consistency for set-valued dynamic risk measurse for bounded
  discrete-time processes.
\newblock {\em Mathematics and Financial Economics}, 12(3):305--333, 2018.

\bibitem{CDK06}
Patrick Cheridito, Freddy Delbaen, and Michael Kupper.
\newblock Dynamic monetary risk measures for bounded discrete-time processes.
\newblock {\em Electronic Journal of Probability}, 11(3):57--106, 2006.

\bibitem{CK10}
Patrick Cheridito and Michael Kupper.
\newblock Composition of time-consistent dynamic monetary risk measures in
  discrete time.
\newblock {\em International Journal of Theoretical and Applied Finance},
  14(1):137--162, 2011.

\bibitem{CS09}
Patrick Cheridito and Mitja Stadje.
\newblock Time-inconsistency of {V}a{R} and time-consistent alternatives.
\newblock {\em Finance Research Letters}, 6(1):40--46, 2009.

\bibitem{christensen2018finding}
S{\"o}ren Christensen and Kristoffer Lindensj{\"o}.
\newblock On finding equilibrium stopping times for time-inconsistent markovian
  problems.
\newblock {\em SIAM Journal on Control and Optimization}, 56(6):4228--4255,
  2018.

\bibitem{D02}
Freddy Delbaen.
\newblock Coherent risk measures on general probability spaces.
\newblock In {\em Advances in finance and stochastics}, pages 1--37, Berlin,
  2002. Springer.

\bibitem{D06}
Freddy Delbaen.
\newblock The structure of m-stable sets and in particular of the set of risk
  neutral measures.
\newblock In Michel \'{E}mery and Marc Yor, editors, {\em In Memoriam
  Paul-Andr\'{e} Meyer}, volume 1874 of {\em Lecture Notes in Mathematics},
  pages 215--258. Springer Berlin / Heidelberg, 2006.

\bibitem{DPRG10}
Freddy Delbaen, Shige Peng, and Emanuela~Rosazza Gianin.
\newblock Representation of the penalty term of dynamic concave utilities.
\newblock {\em Finance and Stochastics}, 14(3):449--472, 2010.

\bibitem{DS05}
Kai Detlefsen and Giacomo Scandolo.
\newblock Conditional and dynamic convex risk measures.
\newblock {\em Finance and Stochastics}, 9(4):539--561, 2005.

\bibitem{ekeland2006being}
Ivar Ekeland and Ali Lazrak.
\newblock Being serious about non-commitment: subgame perfect equilibrium in
  continuous time.
\newblock {\em Preprint}, 2006.

\bibitem{ekeland2008investment}
Ivar Ekeland and Traian~A Pirvu.
\newblock Investment and consumption without commitment.
\newblock {\em Mathematics and Financial Economics}, 2(1):57--86, 2008.

\bibitem{FMM13}
Walter Farkas, Pablo Koch-Medina, and Cosimo Munari.
\newblock Measuring risk with multiple eligible assets.
\newblock {\em Mathematics and Financial Economics}, pages 1--25, 2014.

\bibitem{FR12}
Zachary Feinstein and Birgit Rudloff.
\newblock Time consistency of dynamic risk measures in markets with transaction
  costs.
\newblock {\em Quantitative Finance}, 13(9):1473--1489, 2013.

\bibitem{FR13-survey}
Zachary Feinstein and Birgit Rudloff.
\newblock A comparison of techniques for dynamic multivariate risk measures.
\newblock In A.~Hamel, F.~Heyde, A.~L{\"o}hne, B.~Rudloff, and C.~Schrage,
  editors, {\em Set Optimization and Applications in Finance}, PROMS series,
  pages 3--41. Springer, 2015.

\bibitem{FR12b}
Zachary Feinstein and Birgit Rudloff.
\newblock Multi-portfolio time consistency for set-valued convex and coherent
  risk measures.
\newblock {\em Finance and Stochastics}, 19(1):67--107, 2015.

\bibitem{FR14-alg}
Zachary Feinstein and Birgit Rudloff.
\newblock A recursive algorithm for multivariate risk measures and a set-valued
  {B}ellman's principle.
\newblock {\em Journal of Global Optimization}, 68(1):47--69, 2017.

\bibitem{FR15-supermtg}
Zachary Feinstein and Birgit Rudloff.
\newblock A supermartingale relation for multivariate risk measures.
\newblock {\em Quantitative Finance}, 18(12):1971--1990, 2018.

\bibitem{FR18-scalar}
Zachary Feinstein and Birgit Rudloff.
\newblock Time consistency for scalar multivariate risk measures.
\newblock {\em Preprint}, 2019.

\bibitem{FRW15}
Zachary Feinstein, Birgit Rudloff, and Stefan Weber.
\newblock Measures of systemic risk.
\newblock {\em SIAM Journal on Financial Mathematics}, 8(1):672--708, 2017.

\bibitem{FRZ20}
Zachary Feinstein, Birgit Rudloff, and Jianfeng Zhang.
\newblock Dynamic set values for nonzero sum games with multiple equilibriums.
\newblock {\em Preprint}, 2020.

\bibitem{FP06}
Hans F\"{o}llmer and Irina Penner.
\newblock Convex risk measures and the dynamics of their penalty functions.
\newblock {\em Statistics and decisions}, 24(1):61--96, 2006.

\bibitem{FS02}
Hans F{\"o}llmer and Alexander Schied.
\newblock Convex measures of risk and trading constraints.
\newblock {\em Finance and Stochastics}, 6(4):429--447, 2002.

\bibitem{FS11}
Hans F{\"o}llmer and Alexander Schied.
\newblock {\em Stochastic finance}.
\newblock Walter de Gruyter \& Co., Berlin, third edition, 2011.

\bibitem{FG02}
Marco Frittelli and Emanuela Rosazza~Gianin.
\newblock Putting order in risk measures.
\newblock {\em Journal of Banking \& Finance}, 26(7):1473--1486, 2002.

\bibitem{FG04}
Marco Frittelli and Emanuela Rosazza~Gianin.
\newblock Dynamic convex risk measures.
\newblock In G.P. Szeg\"{o}, editor, {\em New Risk Measures for the 21th
  Century}, pages 227--248. John Wiley \& Sons, 2004.

\bibitem{FS06}
Marco Frittelli and Giacomo Scandolo.
\newblock Risk measures and capital requirements for processes.
\newblock {\em Mathematical Finance}, 16(4):589--612, 2006.

\bibitem{H09}
Andreas~H. Hamel.
\newblock A duality theory for set-valued functions {I}: Fenchel conjugation
  theory.
\newblock {\em Set-Valued and Variational Analysis}, 17(2):153--182, 2009.

\bibitem{HH10}
Andreas~H. Hamel and Frank Heyde.
\newblock Duality for set-valued measures of risk.
\newblock {\em SIAM Journal on Financial Mathematics}, 1(1):66--95, 2010.

\bibitem{HHH07}
Andreas~H. Hamel, Frank Heyde, and Markus H{\"o}hne.
\newblock {\em Set Valued Measures of Risk}.
\newblock Report. Inst. f{\"u}r Mathematik, 2007.

\bibitem{HHR10}
Andreas~H. Hamel, Frank Heyde, and Birgit Rudloff.
\newblock Set-valued risk measures for conical market models.
\newblock {\em Mathematics and Financial Economics}, 5(1):1--28, 2011.

\bibitem{HR08}
Andreas~H. Hamel and Birgit Rudloff.
\newblock Continuity and finite-valuedness of set-valued risk measures.
\newblock In C.~Tammer and F.~Heyde, editors, {\em Festschrift in Celebration
  of Prof. Dr. Wilfried Grecksch's 60th Birthday}, pages 46--64. Shaker Verlag,
  2008.

\bibitem{HMBS16}
Hannes Hoffmann, Thilo Meyer-Brandis, and Gregor Svindland.
\newblock Strongly consistent multivariate conditional risk measures.
\newblock {\em Mathematics and Financial Economics}, 12(3):413--444, 2018.

\bibitem{huang2018time}
Yu-Jui Huang and Adrien Nguyen-Huu.
\newblock Time-consistent stopping under decreasing impatience.
\newblock {\em Finance and Stochastics}, 22(1):69--95, 2018.

\bibitem{JK95}
Elyes Jouini and Hedi Kallal.
\newblock Martingales and arbitrage in securities markets with transaction
  costs.
\newblock {\em Journal of Economic Theory}, 66(1):178--197, 1995.

\bibitem{JMT04}
Elyes Jouini, Moncef Meddeb, and Nizar Touzi.
\newblock Vector-valued coherent risk measures.
\newblock {\em Finance and Stochastics}, 8(4):531--552, 2004.

\bibitem{K99}
Yuri~M. Kabanov.
\newblock Hedging and liquidation under transaction costs in currency markets.
\newblock {\em Finance and Stochastics}, 3(2):237--248, 1999.

\bibitem{KS09}
Yuri~M. Kabanov and Mher Safarian.
\newblock {\em {Markets with Transaction Costs: Mathematical Theory}}.
\newblock Springer Finance. Springer, 2009.

\bibitem{KMZ17}
Chandrasekhar Karnam, Jin Ma, and Jianfeng Zhang.
\newblock Dynamic approaches for some time-inconsistent optimization problems.
\newblock {\em The Annals of Applied Probability}, 27(6):3435--3477, 12 2017.

\bibitem{KS07}
Susanne Kl\"{o}ppel and Martin Schweizer.
\newblock Dynamic indifference valuation via convex risk measures.
\newblock {\em Mathematical Finance}, 17(4):599--627, 2007.

\bibitem{K09}
Christos~E. Kountzakis.
\newblock Generalized coherent risk measures.
\newblock {\em Applied Mathematical Sciences}, 3(49):2437--2451, 2009.

\bibitem{KR18}
Gabriela Kov{\'a}{\v c}ov{\'a} and Birgit Rudloff.
\newblock Time consistency of the mean-risk problem.
\newblock {\em Operations Research}, 2020.
\newblock Forthcoming.

\bibitem{KOZ15}
Eduard Kromer, Ludger Overbeck, and Katrin Zilch.
\newblock Dynamic systemic risk measures for bounded discrete-time processes.
\newblock {\em Mathematical Methods of Operations Research}, 90(1):77--108,
  2019.

\bibitem{KOZ16}
Eduard Kromer, Ludger Overbeck, and Konrad Zilch.
\newblock Systemic risk measures on general measurable spaces.
\newblock {\em Mathematical Methods of Operations Research}, 84(2):323--357,
  2016.

\bibitem{lacker2018law}
Daniel Lacker.
\newblock Law invariant risk measures and information divergences.
\newblock {\em Dependence Modeling}, 6(1):228--258, 2018.

\bibitem{lacker2018liquidity}
Daniel Lacker.
\newblock Liquidity, risk measures, and concentration of measure.
\newblock {\em Mathematics of Operations Research}, 43(3):813--837, 2018.

\bibitem{LR11}
Andreas {L{\"o}hne} and Birgit Rudloff.
\newblock {An algorithm for calculating the set of superhedging portfolios in
  markets with transaction costs}.
\newblock {\em International Journal of Theoretical and Applied Finance},
  17(2):1450012, 2014.

\bibitem{M05}
Ilya Molchanov.
\newblock {\em Theory of Random Sets}.
\newblock Probability and Its Applications. Springer, 2005.

\bibitem{PP10}
Teemu Pennanen and Irina Penner.
\newblock Hedging of claims with physical delivery under convex transaction
  costs.
\newblock {\em SIAM Journal on Financial Mathematics}, 1(1):158--178, 2010.

\bibitem{PL97}
Stylianos Perrakis and Jean Lefoll.
\newblock Derivative asset pricing with transaction costs: An extension.
\newblock {\em Computational Economics}, 10(4):359--376, 1997.

\bibitem{R04}
Frank Riedel.
\newblock {Dynamic coherent risk measures}.
\newblock {\em Stochastic Processes and Their Applications}, 112(2):185--200,
  2004.

\bibitem{roorda2007time}
Berend Roorda and J.M. Schumacher.
\newblock Time consistency conditions for acceptability measures, with an
  application to tail value at risk.
\newblock {\em Insurance: Mathematics and Economics}, 40(2):209--230, 2007.

\bibitem{R08}
Alet Roux.
\newblock {\em Options under transaction costs: Algorithms for pricing and
  hedging of {E}uropean and {A}merican options under proportional transaction
  costs and different borrowing and lending rates}.
\newblock VDM Verlag, 2008.

\bibitem{RZ11}
Alet Roux and Tomasz Zastawniak.
\newblock {A}merican and {B}ermudan options in currency markets under
  proportional transaction costs.
\newblock {\em Acta Applicandae Mathematicae}, 141(1):187--225, 2016.

\bibitem{RS05}
Andrzej Ruszczy\'{n}ski and Alexander Shapiro.
\newblock Conditional risk mappings.
\newblock {\em Mathematics of Operations Research}, 31(3):544--561, 2006.

\bibitem{Sc04}
Giacomo Scandolo.
\newblock Models of capital requirements in static and dynamic settings.
\newblock {\em Economic Notes}, 33(3):415--435, 2004.

\bibitem{S04}
Walter Schachermayer.
\newblock The fundamental theorem of asset pricing under proportional
  transaction costs in finite discrete time.
\newblock {\em Mathematical Finance}, 14(1):19--48, 2004.

\bibitem{strotz1955myopia}
Robert~Henry Strotz.
\newblock Myopia and inconsistency in dynamic utility maximization.
\newblock {\em The review of economic studies}, 23(3):165--180, 1955.

\bibitem{tutsch2008update}
Sina Tutsch.
\newblock Update rules for convex risk measures.
\newblock {\em Quantitative Finance}, 8(8):833--843, 2008.

\bibitem{weber2006distribution}
Stefan Weber.
\newblock Distribution-invariant risk measures, information, and dynamic
  consistency.
\newblock {\em Mathematical Finance: An International Journal of Mathematics,
  Statistics and Financial Economics}, 16(2):419--441, 2006.

\bibitem{WA13}
Stefan Weber, William Anderson, Anna-Maria Hamm, Thomas Knispel, Maren Liese,
  and Thomas Salfeld.
\newblock Liquidity-adjusted risk measures.
\newblock {\em Mathematics and Financial Economics}, 7(1):69--91, 2013.

\bibitem{WH14}
Linxiao Wei and Yijun Hu.
\newblock Coherent and convex risk measures for portfolios with applications.
\newblock {\em Statistics \& Probability Letters}, 90:114--120, 2014.

\end{thebibliography}
\end{document}